%% file: cip.tex
\newif\ifsubmit                 %
\submittrue

\ifsubmit                       %
\documentclass[11pt]{article}   %
\else                           %
\documentclass{article}   %
\fi

\usepackage{article} %
\usepackage{math} %
\usepackage{algo} %
\usepackage{wrapfig} %
\usepackage{placeins} %
\usepackage{graphicx} %
\usepackage[longnamesfirst,numbers]{natbib} %
\usepackage{setspace}
\usepackage{pifont}

\usepackage[font={small},labelfont=bf,textfont=it]{caption}

\ifsubmit \usepackage[letterpaper]{geometry}
\fi

\allowdisplaybreaks

\everymath{\displaystyle}

\newcommand{\setcover}{\text{Set Cover}\xspace} %
\newcommand{\cip}{\text{CIP}\xspace} %
\newcommand{\cipi}{\ensuremath{\text{CIP}_\infty}\xspace} %
\newcommand{\kncover}{\text{Knapsack Cover}\xspace} %
\providecommand{\cost}{c\optsub} %
\providecommand{\capacity}{\cost}  %
\providecommand{\weight}{w\optsub} 
\providecommand{\excess}{\hat{b}\optsub} %

\providecommand{\kcost}{\operatorname{cost}\optpar} %
\providecommand{\ksize}{\operatorname{size}\optpar} %
\providecommand{\kcosts}{\mathcal{C}\optsub} %
\providecommand{\kratios}{\mathcal{R}\optsub} %
\providecommand{\apxkcost}{\widetilde{\operatorname{cost}}\optpar} %
\newcommand{\leaves}{\mathcal{L}} %
 \newcommand{\varj}{\hat{\jmath}}
\newcommand{\alphas}{\mathcal{A}}
\providecommand{\apxpm}{\parof{1 \pm \bigO{1}}}

\providecommand{\citet}[1]{?? \cite{#1}}
\providecommand{\citep}{\cite} 

\newcommand{\cflp}{Carr et al.\ \cite{cflp-00}\xspace}
\newcommand{\ky}{Kolliopoulos and Young \cite{ky-05}\xspace}

\newcommand{\basiclp}{\text{Basic-LP}\xspace}
\newcommand{\score}{\rho\optsub} %
\newcommand{\columncount}{\Delta_0} %
\newcommand{\columnsum}{\Delta_1} %
\newcommand{\rowcount}{\Gamma_0} %
\newcommand{\rowsum}{\Gamma_1} %

\newcommand{\bmin}{b_{\min}}

\newcommand{\pess}{\Phi\optpar} %
\newcommand{\pessrow}[1]{\psi_{#1}\optpar} %
\newcommand{\pessrowA}[1]{\psi_{#1}'\optpar} %
\newcommand{\pessrowB}[1]{\psi_{#1}''\optpar}

\newcommand{\floor}[1]{\lfloor #1 \rfloor}
\newcommand{\ceil}[1]{\lceil #1 \rceil}

\newcommand{\N}{\mathcal{N}}
\newcommand{\M}{\mathcal{M}}

\begin{document}

\title{On Approximating (Sparse) Covering Integer Programs\footnote{This work
    is partially supported by NSF grant CCF-1526799. University of
    Illinois, Urbana-Champaign, IL 61801. {\tt
      \{chekuri,quanrud2\}@illinois.edu}.}}

\author{Chandra Chekuri \and Kent Quanrud}

\maketitle

\begin{abstract}
  We consider approximation algorithms for covering integer programs
  of the form
  \begin{math}
    \text{min } \rip{\cost}{x} \text{ over } x \in \nnintegers^n
    \text{ s.t.\ } A x \geq b \text{ and } x \leq d;
  \end{math}
  where $A \in \nnreals^{m \times n}$,
  $b \in \nnreals^m$, and $\cost, d \in \nnreals^n$ all have
  nonnegative entries. We refer to this problem as \cip, and the
  special case without the multiplicity constraints $x \le d$ as
  \cipi. These problems generalize the well-studied
  \setcover problem. We make two algorithmic contributions.

  First, we show that a simple algorithm based on randomized rounding
  with alteration improves or matches the best known approximation
  algorithms for \cip and \cipi in a wide range of parameter settings,
  and these bounds are essentially optimal. As a byproduct of the
  simplicity of the alteration algorithm and analysis, we can
  derandomize the algorithm without any loss in the approximation
  guarantee or efficiency. Previous work by Chen, Harris and
  Srinivasan \cite{chs-16} which obtained near-tight bounds is based
  on a resampling-based randomized algorithm whose analysis is
  complex.

  Non-trivial approximation algorithms for \cip are based on solving
  the natural LP relaxation strengthened with \emph{knapsack cover}
  (KC) inequalities \cite{cflp-00,ky-05,chs-16}. Our second
  contribution is a fast (essentially near-linear time) approximation
  scheme for solving the strengthened LP with a factor of $n$ speed up
  over the previous best running time \cite{cflp-00}.  To achieve this
  fast algorithm we combine recent work on accelerating the
  multiplicative weight update framework with a partially dynamic
  approach to the knapsack covering problem.

  Together, our contributions lead to near-optimal (deterministic)
  approximation bounds with near-linear running times for \cip and
  \cipi.
\end{abstract}

\thispagestyle{empty}

\newpage

\setcounter{page}{1}

\section{Introduction}
\setcover is a fundamental problem in discrete optimization with many
applications and connections to other problems.  A number of variants
and generalization of \setcover have been studied over the years. In
this paper we consider a general problem that captures many of these
as special cases. This is the minimum cost covering integer program
problem (\cip for short). It is a class of integer programs of the
form
\begin{align*}
  \text{minimize }
  \rip{\cost}{x}
  \text{ over }
  x \in \integers_{\ge 0}^n
  \text{ s.t.\ }
  A x \geq b
  \text{ and }
  x \leq d,
  \labelthisequation[\cip]{mc}
\end{align*}
where $A \in \nnreals^{m \times n}$, $b \in \nnreals^m$, and
$\cost, d \in \nnreals^n$ all have nonnegative coefficients. We let
$N = \norm{A}_0 + \norm{b}_0 + \norm{\cost}_0 + \norm{d}_0 > m + n$
denote the total number of nonzeroes in the input. Let $C \geq 1$ be such
that $A_{i,j} \in \setof{0} \cup [1/C, C]$ for all $i$ and $j$;
$\log C$ reflects the number of bits required to write down a
coefficient of $A_{i,j}$.  Since we are interested in integer
solutions we can assume, without loss of generality, that
$d \in \integers^n$ and $A_{i,j} \leq b_i$ for all $i,j$.
An important special case of \cip is when there are no
multiplicity constraints, in other words, $d_j = \infty$ for all $j$.
We refer to this problem as \cipi.

\cip and \cipi can be understood combinatorially as multiset
multicover problems, particularly if we assume for convenience that
$A$ and $b$ have integer entries. The elements that need to be covered
correspond to the rows of $A$, (say) the element $e_i$ for row $i$,
for a total of $m$ elements. For each element $e_i$, $b_i$ is the
requirement on the number of times $e_i$ needs to be covered.  Each
column $j$ of $A$ correspond to a multiset $S_j$. For each element
$e_i$ and multiset $S_j$, $A_{i,j}$ is the number of times $S_j$
covers $e_i$. Multiplicity constraints are specified by $d$; $d_j$ is
the maximum number of copies of $S_j$ that can be chosen. The cost of
one copy of $S_j$ is $c_j$. The goal is to pick a minimum cost
collection of multisets (with copies allowed) that together cover all the
requirements of the elements, while respecting the multiplicity bounds
on the sets.

\setcover is a special case of \cipi where $A$ is a
$\setof{0,1}$-matrix and $b = \ones$; $A_{i,j} = 1$ implies that
element $e_i$ is in set $S_j$. In this setting the multiplicity bounds
are irrelevant since at most one copy of a set is ever
needed. \setcover is NP-Hard and its approximability has been
extensively studied.  A simple greedy algorithm achieves an
approximation of $H_k \le (1 + \ln k) \le (1+\ln m)$ where
$k = \max_i |S_i|$ is the maximum set size \cite{v-13,ws-11}; here
$H_k = 1 + 1/2 + \ldots + 1/k$ is the $k$'th harmonic
number\footnote{We will be concerned with the setting of arbitrary
  costs for the sets. Unit-cost \setcover admits some improved bounds
  and those will not be the focus here.}.  Unles $P = NP$ there is no
$(1-\delta) \ln m$ approximation for any fixed $\delta > 0$
\cite{m-15,feige-98} where $m$ is the number of elements\footnote{It
  is common to use $n$ for the number of elements and $m$ for the
  number of sets. However, in the setting of covering integer
  programs, it is natural to use our notation in accordance with the
  usual notation for optimization where $n$ is the number of decision
  variables and $m$ is the number of constraints. See \cite{ky-05}
  among others.}. Another approximation bound for \setcover is $f$
where $f$ is the maximum frequency \cite{h-82}; the frequency of an
element is the number of sets that contain it. This bound is achieved
via the natural LP relaxation.  For any fixed $f$, \setcover instances
with maximum frequency $f$ are hard to approximate to within a
$f-\eps$ factor under UGC \cite{bk-10}, and to within $f-1-\eps$ under
$P \neq NP$ \cite{dgkr-05}.

The greedy algorithm for \setcover can be easily generalized to
\cip. Dobson \cite{dobson} analyzed this extension and showed that,
when all entries of $A,b$ are \emph{integers}, it has an approximation
bound of $H_d$, where
$d = \max_{1 \le j \le n} \sum_{i} A_{i,j} \geq C$ is the maximum
column sum\footnote{Wolsey's analysis \cite{wolsey} for the Submodular
  Set Cover problem further generalizes Dobson's result.}. The
analysis for the greedy bound is tight, and the dependency on the
maximum coordinate $C$ is undesirable for several reasons.  In
particular, the entries in $A$ and $b$ can be rational, and the greedy
algorithm's approximation ratio can be as large as $m$ \cite{dobson}
\footnote{Dobson also described a variant of greedy for rational data
  whose approximation ratio is
  $1+\max_j (H_{d_j} + \ln \sum_{i}A_{i,j})$; $d_j$ is the number of
  non-zeroes in column $j$ and the entries of $A$ are assumed to be
  scaled such that the minimum non-zero entry of each row is at least
  $1$.}.  The question of obtaining an improved approximation ratio
that did not depend on $C$ was raised in \cite{dobson}.  For \cipi,
when there are no multiplicity constraints, \citeauthor{rt-87}, in
their influential work on randomized rounding, used the LP relaxation
of \refequation{mc} (which we refer to as \basiclp) to obtain an
$O(\log m)$ approximation \cite{rt-87}. Subsequent work has refined
and improved this bound, and later we will describe recent
approximation bounds by \citet{chs-16} that are much tighter w/r/t the
sparsity of a given instance.

\paragraph{Stronger LP Relaxation:}
In the presence of multiplicity constraints, \basiclp has an unbounded
integrality gap even when $m=1$, which corresponds to the \kncover
problem. The input to this problem consists of $n$ items with item $i$
having cost $c_i$ and size $a_i$, and the goal is to find a minimum
cost subset of the items whose total size is at least a given quantity
$b$. To illustrate the integrality gap of the LP relaxation consider
the following simple example from \cite{cflp-00}.
\begin{align*}
  \min x_1                      %
  \text{ over } x_1,x_2 \geq 0 %
  \text{ s.t.\ } Bx_1 + (B-1)x_2 \ge B                         %
  \text{ and }  x_1, x_2 \le 1.
\end{align*}
It is easy to see that the optimum integer solution has value $1$
while the LP relaxation has value $1/B$, leading to an integrality gap
of $B$. The example shows that the integrality gap is large even when
$d = \ones$, a natural and important setting.

To overcome this gap, \cflp suggested the use of knapsack cover (KC)
inequalities to strengthen the LP.  We describe the idea. For each
$S \subseteq [n]$, one can consider the residual covering constraints
if we force $x_i = d_i$ for all $i \in S$. The residual system is
called the \emph{knapsack covering constraint} for $S$ and written as
$A_S x \geq b_S$, where $b_S \in \reals^n$ is defined by
\begin{math}
  b_{S,i} = \max{0,b_i - \sum_{j \in S}A_{i,j}d_j},
\end{math}
and for $j \in [m]$, $A_{S,i,j}$ is defined by
\begin{align*}
  A_{S,i,j} =
  \begin{cases}
    0 &\text{if } j \in S, \\
    \min{A_{i,j}, b_{S,i}} &\text{otherwise}.
  \end{cases}
\end{align*}
That is, for $S \subseteq [n]$, we compute the residual demand
$b_{S}$, zero the coefficients in $A$ of any contracted coordinate
$j \in S$, and reduce each remaining coefficient in $A$ to be at most
the residual covering demand.

A feasible integral solution $x$ to \refequation{mc} satisfies
$A_S x \geq b_S$ for all $S \subseteq [n]$. The following LP, then, is
a valid linear relaxation of the integer program \refequation{mc}.
\begin{align*}
  \text{minimize }
  \rip{\cost}{x}
  \text{ over }
  x \in \nnreals^n
  \text{ s.t.\ }
  x \leq d
  \text{ and }
  A_S x \geq b_S
  \text{ for all }
  S \subseteq [n].
  \labelthisequation[KC-LP]{kc}
\end{align*}
Note that the knapsack cover constraints made the packing constraints
$x \leq d$ redundant and expendable; \refequation{kc} is a pure
covering problem.  Given a feasible solution $x$ to \refequation{kc},
one can (randomly) round $x$ to a feasible integer solution $y$ to
\refequation{mc}. Kolliopoulos and Young \cite{ky-05} obtain an
$O(\log \columncount)$ approximation via \refequation{kc} where
$\columncount$ is the maximum number of non-zeroes in any column of
$A$; note that $\columncount \le m$. Recent tighter bounds
\cite{chs-16} will be discussed shortly.

\paragraph{Sparsity bounds and motivation:}
We are motivated by the following high-level question. Can one obtain
near-tight approximation bounds for \cip and \cipi that are efficient,
simple and deterministic? \setcover is the model here where a simple
greedy algorithm or a simple primal-dual algorithm gives provably
optimal worst-case approximation ratios with near-linear running time.
We briefly discuss some existing results before stating our results.

The first issue is regarding the approximation ratios for \cip and
\cipi. Currently the best bounds in terms of column sparsity are from
the recent work of \citet*{chs-16}. To
describe the bounds, we assume, without loss of generality, that the
problem is normalized such that entries of $A$ are in $[0,1]$ and
$b \ge \ones$. Following \cite{chs-16}, we let $\columncount$ denote
the maximum number of non-zeroes in any column of $A$, and let
\begin{math}
  \columnsum = \max_{j \in [n]} \sum_{i=1}^m A_{i,j}
\end{math}
denote the maximum column sum. $\columncount$ and $\columnsum$ are the
$\ell_0$ and $\ell_1$ measures of column sparsity of $A$. Similarly we
let $\rowcount$ and $\rowsum$ denote the corresponding measures for
row sparsity of $A$. In the context of \setcover, $\columncount$ is
the maximum set size and $\rowcount$ is the maximum frequency.  Under
the normalization\footnote{Note that the result of Dobson is based on
  a very different normalization.} that
$A_{i,j} \in [0,1]^{m \times n}$, we have
$\columnsum \le \columncount \le m$, and in some cases
$\columnsum \ll \columncount$.  We let $\bmin$ denote
$\min_{i \in [m]} b_i$, and this measures the so-called ``width'' of
the system. As $\bmin$ increases, the problem gets easier in the case
of \cipi. We summarize the relevant approximation ratios from
\cite{chs-16}, all of which are randomized\footnote{These bounds were
  based on the latest version of \cite{chs-16} at the time of our
  work.  
  \citeauthor{chs-18-arxiv} have improved their bounds in a recently
  updated preprint \citep{chs-18-arxiv}, essentially replacing the
  second order $\sqrt{\cdots}$ terms with $\bigO{\ln \ln{\cdots}}$
  terms.  We plan to do a careful comparison with the results in
  \cite{chs-18-arxiv} in a future version of this work.}.
\begin{itemize}
\item A
  $\parof{1 + \ln \columncount + \bigO{\sqrt{\ln
        \columncount}}}$-approximation for \cip via \refequation{kc}.
\item A
  \begin{math}
    \parof{1 + \frac{\ln (1+\columnsum)}{\bmin} + 4\sqrt{ \frac{\ln
          (1+\columnsum)}{\bmin}}}
  \end{math}
  approximation for \cipi via \basiclp.
\item A bicriteria algorithm for \cip: given $\eps > 0$, the algorithm
  outputs an integer solution $z$ with cost at most
  \begin{math}
    \parof{1+ %
      4\frac{\ln (1+\columnsum)}{\bmin \cdot \eps} + %
      5\sqrt{ \frac{\ln (1+\columnsum)}{\bmin \cdot \eps}} %
    }                                                 %
  \end{math}
  times the LP value, and satisfies the multiplicity constraints to
  within a $(1+\eps)$-factor, that is $z \le \ceil{(1+\eps)d}$.
\end{itemize}
The algorithmic framework of \cite{chs-16} is based on resampling,
which is motivated by the developments on the constructive version of
the Lov\'asz Local Lemma starting with the work of Moser and Tardos
and continuing through several subsequent developments. Although the
high-level algorithmic idea is not that complicated, the analysis is
technically involved and randomization seems inherently necessary. One
of the significant and novel contributions of \cite{chs-16} is to
show, for the first time, that approximation bounds based on
$\columnsum$ are feasible.  Not only can $\columnsum$ be much smaller
than $\columncount$, it is also more robust to noise and
perturbation. Noise is typically not an issue for combinatorial
instances such as those arising in \setcover but can be relevant in
instances of \cip that arise from data with real numbers.  We note
that there are two regimes of interest for $\columnsum$. One regime is
when $\columnsum$ is large (at least some fixed constant) in which
case the approximation bounds tend towards $\ln \columnsum$ plus lower
order terms.  The other regime is when $\columnsum$ is small and
tends to $0$; in this regime the approximation ratio guarantees from
preceding bounds for \cipi tend to $1 + O(\sqrt{\columnsum})$.
Approximation bounds in terms of row sparsity are also known for
\cip. Pritchard and Chakrabarty \cite{pc-11} describe a $\rowcount$
approximation for \cip (previous results obtained a similar bound in
more restricted settings).  For \cipi a bound of $(1+\rowsum)$ is
implicit in \cite{pc-11} (see Proposition 7).

The second issue is with regards to efficiency.  A
$(1-\eps)$-approximation for the \basiclp can be obtained in
near-linear $\apxO(N/\eps^2)$ time \cite{young-14} or
$\apxO{N/\eps + m /\eps^2 + n/\eps^3}$ randomized time \cite{cq-18}
(and more efficiently if there are no multiplicity constraints
\cite{ky-efpc-14,wrm-16}). On the other hand, \refequation{kc} is not
as simple to solve because of the exponential number of implicit
constraints. \cflp describe two methods to solve \refequation{kc}.
The first is to use the Ellipsoid method via an approximate separation
oracle\footnote{The LP is not solved exactly but the fractional
  solution output by the algorithm provides a lower bound and suffices
  for the current randomized rounding algorithms.}.  The other is to
use a Lagrangean relaxation based approximation scheme which yields a
running time of $\bigO{n N \poly{\frac{1}{\eps}} \log C}$ for a
$(1+\eps)$-approximate solution.  \ky explicitly raise the question of
a fast approximation algorithm for \cip. The recent work of
\citet{chs-16} discussed above shows that a randomized rounding
technique via a resampling framework yields near-optimal
approximations, and they run in expected near-linear time. This is in
contrast to some previous rounding algorithms \cite{s-99,s-06} that
were rather complex and slow. Hence the bottleneck for \cip is solving
the LP relaxation \refequation{kc}.

\subsection{Our Results}
In this paper we address both the approximability and efficiency of
\cip and \cipi.

Our first set of results is on rounding the fractional solution to the
\basiclp and \refequation{kc}. Our main contribution is to show that a
very simple combination of randomized rounding followed by alteration,
that has been previously considered for covering integer programs
\cite{s-01,ss-12,gn-14}, yields clean and (in some cases) improved
bounds when compared to those of \cite{chs-16}. Our focus in this
paper is in the regime where $\columncount$ and $\columnsum$ are
larger than some fixed (modest) constant. Under this assumption, we
obtain the following improved approximation bounds:
\begin{itemize}
\item A
  \begin{math}
    \parof{\ln \columncount + \ln \ln \columncount + \bigO{1}}
  \end{math}-approximation for \cip via \refequation{kc}.
\item A \begin{math}
    \parof{\ln \columnsum + \ln \ln \columnsum + O(1)}
  \end{math}-approximation for \cipi via \basiclp. When $\bmin$ is
  large
  the ratio improves to \begin{math}
    \parof{\frac{\ln \columnsum}{\bmin} + \ln \frac{\ln
        \columnsum}{\bmin} + O(1)}
  \end{math} under the assumption that $\frac{\ln \columnsum}{\bmin}$
  is sufficiently large.
\item A bicriteria algorithm for \cip via \refequation{kc}: given an
  error parameter $\eps > 0$, the algorithm outputs a solution with
  cost at most
  \begin{math}
    \parof{\frac{\ln \columnsum}{\bmin} + \ln \frac{\ln
        \columnsum}{\bmin} + O\parof{\ln \reps}}
  \end{math}
  times the LP value, and satisfies the multiplicity constraints to
  within a $(1+\eps)$-factor\footnote{It is not hard to show that
    bounds based on $\columnsum$ are not feasible for \cip if
    multiplicity constraints are not violated.}.
\end{itemize}

When $\columncount$ is large, \cite{chs-16} established a hardness
lower bound of the form $\ln \columncount - c \ln \ln \columncount$
for some constant $c$ by extending a result of \citet{t-01} for
\setcover. Thus our improved bounds in this regime get closer to the
lower bound in the second order term. We obtain a more substantial
improvement for the bicriteria approximation as a function of $\eps$,
and it is important to observe that it is based on \refequation{kc}
while \cite{chs-16} uses \basiclp. Perhaps of greater interest than
the precise improvements is the fact that the alteration algorithm is
simple and easy to analyze.  All we really need is a careful use of
the lower tail of the standard Chernoff bound. A significant
consequence of the simple analysis is that we are able to easily and
efficiently derandomize the algorithm via the standard method of
conditional expectations. This leads to simple \emph{deterministic}
algorithms without loss in the approximation bounds. We also believe
that our analysis is insightful for the bound based on $\columnsum$;
it is not easy to see why such a bound should be feasible in the first
place. Finally we note that our bound for \cip based on $\columncount$
is quite elementary and direct, and does not rely on the more involved
analysis for the bound based on $\columnsum$; this is not the case in
\cite{chs-16}.

\begin{remark}
  For \cipi when $\columnsum$ is sufficiently small we can show that
  the alteration approach gives an approximation ratio of
  \begin{math}
    \bigO{1 + \bigO{\sqrt{\columnsum \log \frac{1}{\columnsum}}}},
  \end{math}
  which tends to $1$ as $\columnsum \rightarrow 0$. This is slightly
  weaker than the bound from \cite{chs-16}, which gives an
  approximation ratio of $\bigO{1 + O(\sqrt{\columnsum})}$. We believe
  that our analysis is of the alteration algorithm is not tight.
\end{remark}

\begin{remark}
  The alteration based algorithms have a single parameter $\alpha$
  that controls the scaling of the variables in the randomized
  rounding step. Depending on the regime of interest we choose an
  appropriate $\alpha$ to obtain the best theoretical guarantee. One
  can try essentially all possible values of $\alpha$ (by appropriate
  discretization) and take the best solution. This makes the algorithm
  oblivious to the input parameters. The fixing step in the alteration
  algorithm depends on the objective function $c$ and is
  deterministic. One can make the fixing step oblivious to $c$ via
  randomization and known results on the \kncover problem\footnote{In
    retrospect, the algorithms in \cite{chs-16} appear to be
    randomized and oblivious fixing schemes. However, the analysis is
    involved for various technical reasons.}.
\end{remark}

Our second result is a fast approximation scheme for solving
\refequation{kc}, improving upon the previous bound in \cite{cflp-00}
by a factor of $n$. For polynomially-bounded $C$ (which is a
reasonable assumption in many settings) the running time is
near-linear for any fixed $\eps$. The precise result is stated in the
theorem below.  To achieve the result we develop an incremental
dynamic data structure for the \kncover problem and combine it with
other data structures following our recent line of work on speeding up
MWU based approximation schemes for implicit positive linear programs.

\begin{theorem}
  \labeltheorem{kc-ltas} Let $\eps > 0$ be fixed, and consider an
  instance of \refequation{kc} and let $\opt$ be the value of an
  optimum solution. There is a deterministic algorithm that in
  $\apxO{\frac{N \log C}{\eps^3} + \frac{\parof{m + n} \log
      C}{\eps^5}}$ time outputs $x \in \reals^n$ such that
  \begin{math}
    \rip{c}{x} \leq \opt,
  \end{math}
  \begin{math}
    x \leq d,
  \end{math}
  and for all $S \subseteq [n]$, $A_S x \geq \epsless b_S$.
\end{theorem}

Together, our results yield deterministic and fast approximation
algorithms for \cip and \cipi that are near-optimal in a wide range of
parameter settings. Our analysis demonstrates that random rounding
plus alteration provides near-tight bounds for \setcover, \cip and
\cipi. The rounding algorithm and analysis can be considered textbook
material.

\subsection{Techniques and other related work}
Approximation algorithms and hardness results for \setcover, its
generalizations, and important special cases have been extensively
studied in the literature. We refer the reader to approximation books
for standard and well-known results \cite{v-13,ws-11}.  \cip and \cipi
have been primarily addressed in previous work via LP relaxations and
randomized rounding starting with the well-known work of Raghavan and
Thompson \cite{rt-87}. Srinivasan has used sophisticated probabilistic
techniques based on the Lov\'asz Local Lemma (LLL) and its
derandomization via pessimistic estimators to obtain bounds that
depend on the sparsity of $A$ \cite{s-99,s-06}. These were also used,
in a black box fashion, for \cip by Kolliopoulos and Young
\cite{ky-05}.  The more recent work by \citet{chs-16} is inspired by
the ideas surrounding the Moser-Tardos resampling framework
\cite{mt-10,hs-13} that led to constructive versions of LLL. They were
the first to consider $\ell_1$-sparsity based bounds. Srinivasan
\cite{s-01} used randomized rounding with alteration for covering and
packing problems. For packing problems the alteration approach has led
to a broad framework on contention resolution with several
applications \cite{bkns-12,cjv-14}.  For covering we are also inspired
by the paper of Saket and Sviridenko \cite{ss-12} who described an
alteration based algorithm for \setcover that achieves an
approximation ratio of
$\rowcount (1-e^{\frac{\ln \columncount}{\rowcount-1}})$ that
addresses both column sparsity and row sparsity in a clean and unified
fashion.  The known hardness of approximation result of \setcover that
we already mentioned carries over to \cip when $\columncount$ is the
parameter of interest. \citet{chs-16}, building upon the results for
\setcover, showed near-optimal integrality gaps and hardness results
for \cip and \cipi in various sparsity regimes as a function of the
parameter $\columnsum$. Together these results show that the current
upper bounds for \cip and \cipi in terms of column and row sparsity
are essentially optimal up to lower order terms.

KC inequalities are well-known in integer programming, and since the
work of \cflp, there have been several uses in approximation
algorithms.  MWU based Lagrangean relaxation methods have been
extensively studied to derive FPTAS's for solving packing, covering
and mixed packing and covering linear programs. We refer the reader to
some recent papers \cite{young-14,mrwz-16,cq-17-soda,cq-18} for
pointers. Building upon some earlier work of Young \cite{young-14}, we
have recently demonstrated, via several applications
\cite{cq-17-soda,cq-17-focs,cq-18}, some techniques to speed up MWU
based approximation schemes for a variety of explicit and implicit
problems. The key idea is the use of appropriate data structures that
mesh with the analysis and flexibility of the high-level MWU based
algorithm. There are certain general principles in this approach and
there are problem specific parts. In this paper our technical
contribution is to adapt the FPTAS for the \kncover problem and
make it dynamic in a manner that is suitable to the needs of the MWU
based updates. This allows us to speed up the basic approach outlined
in \cflp and obtain a fast running time. Note that the rounding step
for \cip loses a large factor in the approximation, and hence solving
the LP to high-precision is not the main focus. Moroever, the large
dependence on $\eps$ is mainly due to the FPTAS for knapsack cover.  A
different line of work based on Nesterov's accelerated gradient
descent achieve running time with a better dependence on $O(1/\eps)$
for solving positive LPs \cite{bi-06,ce-05}.  Until recently, the
running time of these algorithms did not have a good dependence on the
combinatorial parameters. The work of Allen-Zhu and Orecchia has
remedied this for explicit pure packing and covering LPs
\cite{ao-15,wrm-16}. However, it is not clear how well these
techniques can be applied to implicit problems like the ones we
consider here.

\paragraph{Organization:} The paper has several technical results, and
broadly consists of two parts. \refsection{rounding} describes
approximation results obtained via randomized rounding plus alteration
and its derandomization.  \refsection{fast} describes our fast
approximation scheme to solve \refequation{kc} and proves
\reftheorem{kc-ltas}. The two parts can be read independently.  We
made the paper modular to keep it both readable and detailed, and
suggest that the reader skip to section of interest rather than read
the paper sequentially.

\section{Randomized Rounding with Alteration}
\labelsection{rounding} In this section we formally describe two
versions of randomized rounding with alteration and analyze it for
\cipi and \cip. The algorithms are simple and have been proposed and
analyzed previously. We analyze them in a tighter fashion, especially
in terms of the $\ell_1$ sparsity of $A$.  Recall that we have
normalized the problem such that all entries of $A$ are in $[0,1]$ and
$b \ge \ones$. It is convenient to simplify the problem further and
assume that $b = \ones$. This can be done by scaling each row $i$ by
$b_i$. This setting captures the essence of the problem and simplifies
notation and the analysis.

As a preliminary step we state some known facts about the \kncover
problem (which is the special case of \cip with $m=1$) in the lemma
below.

\begin{lemma}[\cflp]
  \labellemma{knapsack-cover} Consider an instance of \kncover of the
  form $\min \rip{c}{x}$ subject to $\sum_{j=1}^n a_j x_j \ge 1$,
  $x \le d$, and $x \in \nnintegers^n$.
  \begin{itemize}
  \item Suppose $y$ is a feasible solution to the \basiclp
    relaxation. Then there is an integer solution $z$ in the support of $y$
    such that
    $\rip{c}{z} \le 2 \rip{c}{y}$ and $y_j \le \ceil{2x_j}$.
  \item Suppose $y$ is a feasible solution to \refequation{kc}
    (satisfying knapsack cover inequalities). Then there is an integer
    solution $z$ in the support of $y$ such that
    $\rip{c}{z} \le 2 \rip{c}{y}$ and $z \le d$.
  \end{itemize}
  Given $y$, an integer vector $z$ satisfying the
  stated conditions can be found in near-linear time.
\end{lemma}

The first algorithm we present is called \algo{round-and-fix}, and
pseudocode is given in \reffigure{round-and-fix}. The input consists
of a nonnegative matrix $A \in \nnreals^{m \times n}$, a cost vector
$c \in \nnreals^n$, a nonnegative vector $x \in \nnreals^n$ such that
$A x \geq \ones$, and a positive parameter $\alpha > 0$. The goal is
to output an integer vector $z \in \nnintegers^n$ such that
$A x \geq \ones$ with cost $\rip{c}{z}$ comparable to $\rip{c}{x}$.
The algorithm consists of a \emph{randomized rounding step}, step
\refstep{rounding-loop} in \reffigure{round-and-fix}, and an
\emph{alteration step}, step \refstep{alteration-loop}.

In the randomized rounding step, \algo{round-and-fix} scales up the
fractional solution $x$ by $\alpha$ and independently rounds each
coordinate $i$ to $\floor{\alpha x_i}$ or $\ceil{\alpha x_i}$ such
that the expectation is $\alpha x_i$.  Let $z$ be the random vector
picked in the first step. $z$ may leave several covering constraints
$i \in [m]$ unsatisfied (that is $(Az)_i < 1$).  These constraints are
fixed in the subsequent \emph{alteration step}. In the alteration
step, each unsatisfied covering constraint $i$ is addressed
separately. Let $y^{(i)}$ be an approximately optimum solution to
cover the $i$'th constraint by itself; we can find a constant factor
approximation since this is the knapsack cover problem.  Letting
$U \subseteq [m]$ denote the subset of unsatisfied constraints, we
output the solution $z'$ where
$z'_j = \max\{z_j, \max_{i \in U} y^{(i)}_j\}$ for each $j \in [n]$.

\begin{remark}
  The fixing step for an unsatisfied constraint $i$ can be done more
  carefully. The residual requirement after the randomized step is
  $1 - (Az)_i$, and we can solve a \kncover problem to satisfy
  this requirement. This is crucial when $\columnsum$ is
  small but makes the analysis complex. For the bounds we
  seek in the regime when $\columncount$ and $\columnsum$ are
  sufficiently large constants, it suffices to fix each unsatisfied
  constraint ignoring the contribution from the first step.
\end{remark}

The alterations guarantee that all of the covering constraints are
met. The cost of the solution is a random variable.  Large values of
$\alpha$ decrease the expected cost for alterations but increase the
initial cost of randomized rounding; small values of $\alpha$ decrease
the initial cost of randomized rounding but increase the expected cost
of alterations.  A careful choice of $\alpha$ leads to an
approximation guarantee independent of $m$ and relative to either the
\emph{$\ell_0$-column sparsity of $A$},
\begin{math}
  \columncount \defeq \max_{j \in [n]} \, \sizeof{\setof{i \in [m]
      \where A_{i,j} \neq 0}},
\end{math}
or the \emph{$\ell_1$-column sparsity}
of $A$,
\begin{math}
  \columnsum \defeq \max_{j \in [n]} \sum_{i=1}^m A_{i,j}.
\end{math}

In the sequel we will assume that $\columncount$ and $\columnsum$ are
greater than some fixed constant.  We believe that this is the main
regime of interest. Moreover, the assumption avoids notational and
technical complexity. In \refappendix{small-columnsum} we consider the
case when $\columnsum$ is small.

\begin{figure}
  \begin{framed}
    \input{algos/round-and-fix}
    \caption{An alteration based rounding algorithm for covering
      programs.\labelfigure{round-and-fix}}
  \end{framed}
\end{figure}

\begin{figure}
  \begin{framed}
    \input{algos/contract-round-fix}
    \caption{An alteration based rounding algorithm for covering
      programs with knapsack covering
      constraints. \labelfigure{contract-round-fix}}
  \end{framed}
\end{figure}

We now state the specific theorems for \cip and \cipi and prove them
in subsequent sections.  The first theorem analyzes the performance of
\algo{round-and-fix} in terms of $\ell_1$ sparsity of $A$. Note that
the theorem is directly relevant for \cipi. The bounds on the
multiplicities in \reftheorem{l1-alteration} are insufficient for
\cip, but will prove useful for \cip when discussing the second
algorithm later.

\begin{theorem}
  \labeltheorem{l1-alteration} Let $A \in [0,1]^{m \times n}$,
  $c \in \nnreals^n$, and $x \in \nnreals^n$ such that
  $A x \geq \ones$. Let
  \begin{math}
    \alpha = \ln \columnsum + \ln \ln \columnsum + \bigO{1}.
  \end{math}
  In $\bigO{\norm{A}_0}$ time,
  \algo{round-and-fix($A$,$c$,$x$,$\alpha$)} returns a randomized
  integral vector $z \in \nnintegers^n$ with coverage
  \begin{math}
    A z \geq \ones,
  \end{math}
  expected cost
  \begin{math}
    \evof{\rip{c}{z}} %
    \leq %
    \parof{\alpha + \bigO{1}} \rip{c}{x},
  \end{math}
  and multiplicities
  \begin{math}
    z < \ceil{\alpha x}.
  \end{math}
\end{theorem}

The following theorem is an easy corollary of the preceding theorem
since $\columnsum \le \columncount$ when entries of $A$ are from
$[0,1]$. Nevertheless we state it separately since its proof is
simpler and more direct and there is a tighter bound on the additive
constant. We encourage a reader new to the alteration analysis to read
the proof of this theorem before that of \reftheorem{l1-alteration}.

\begin{theorem}
  \labeltheorem{l0-alteration} Let $A \in [0,1]^{m \times n}$, $c \in
  \nnreals^n$, and $x \in \nnreals^n$ such that $A x \geq \ones$. Let
  $\alpha = \ln \columncount + \ln \ln \columncount + \bigO{1}$.  In
  $\bigO{\norm{A}_0}$ time, \algo{round-and-fix($A$,$c$,$x$)} returns
  a randomized integral vector $z \in \nnintegers^n$ with coverage
  \begin{math}
    A z \geq \ones,
  \end{math}
  expected cost
  \begin{math}
    \evof{\rip{c}{z}} \leq %
    \parof{\alpha + \bigO{1}} \rip{c}{x},
  \end{math}
  and multiplicities $z \leq \roundup{\alpha x}$.
\end{theorem}

We prove \reftheorem{l0-alteration} in \refsection{l0-alteration} and
we prove \reftheorem{l1-alteration} in \refsection{l1-alteration}.

The upper bounds on the multiplicities, $z < \alpha x + \ones$, are
useful for handling multiplicity constraints via knapsack covering
constraints.  The second algorithm, called \algo{contract-round-fix}
and given in \reffigure{contract-round-fix}, adds a preprocessing step
to \algo{round-and-fix}. The basic idea is by now standard and first
proposed by \cflp to take advantage of the KC inequalities.  At a high
level, we want to simulate the standard randomized rounding without
violating the multiplicity constraints. Note that if
$\alpha x \leq d$, then \algo{round-and-fix} already returns an
integral solution meeting the multiplicity constraints.  If
$\ceil{\alpha x_j} \ge d_j$ for some coordinate $j$, then the
algorithm deterministically sets $z_j = d_j$. After contracting the
large coordinates, the knapsack covering constraints ensure that the
remaining small coordinates still satisfy the residual covering
problem.  We apply \algo{round-and-fix} to the residual
instance.  We prove the following theorems on the performance of
\algo{contract-round-fix} in terms of the $\ell_0$ and $\ell_1$
sparsity of $A$.

\begin{theorem}
  \labeltheorem{l0-alteration-multiplicity} Let
  $A \in [0,1]^{m \times n}$, $c \in \nnreals^n$,
  $d \in \naturalnumbers^n$ and $x \in \nnreals^n$ such that $x$
  covers all the knapsack covering constraints w/r/t $A$ and $d$. Let
  $\alpha = \ln \columncount + \ln \ln \columncount + \bigO{1}$.

  In $\norm{A}_0$ time,
  \algo{contract-round-fix($A$,$c$,$d$,$x$,$\alpha$,0)} returns a
  randomized integral vector $z \in \nnintegers^n$ with coverage
  \begin{math}
    A z \geq \ones,
  \end{math}
  expected cost
  \begin{math}
    \evof{\rip{c}{z}} \leq %
    \parof{\ln \columncount + \ln \ln \columncount + \bigO{1}} %
    \rip{c}{x}
  \end{math}
  and multiplicities
  \begin{math}
    z \leq d.
  \end{math}
\end{theorem}

\begin{theorem}
  \labeltheorem{l1-alteration-multiplicity} Let
  $A \in [0,1]^{m \times n}$, $c \in \nnreals^n$, $\eps \in (0,1]$,
  $d \in \naturalnumbers^n$, and $x \in \nnreals^n$ such that $x$
  covers all the knapsack covering constraints w/r/t $A$ and
  $x \leq d$. Let
  \begin{math}
    \alpha = \ln \columnsum + \ln \ln \columnsum + \bigO{\ln{\reps}}.
  \end{math}

  In $\norm{A}_0$ time,
  \algo{contract-round-fix($A$,$c$,$d$,$x$,$\alpha$,$\eps$)} returns a
  randomized integral vector $z \in \nnintegers^n$ with coverage
  \begin{math}
    \epsmore A z \geq \ones,
  \end{math}
  expected cost
  \begin{math}
    \evof{\rip{c}{z}} %
    \leq %
    \parof{\ln \columnsum + \ln \ln \columnsum + \bigO{\ln{1/\eps}}} %
    \rip{c}{x},
  \end{math}
  and multiplicities
  \begin{math}
    z \leq d.
  \end{math}
\end{theorem}

Scaling up the output to \reftheorem{l1-alteration-multiplicity} by a
$\epsmore$-multiplicative factor and then rounding up to an integral
vectors shifts the approximation error from the coverage constraints
to the multiplicity constraints, as follows.
\begin{corollary}
  Let $A \in [0,1]^{m \times n}$, $c \in \nnreals^n$,
  $\eps \in (0,1]$, $d \in \integers^n$, and $x \in \nnreals^n$ such
  that $x$ covers all the knapsack covering constraints w/r/t $A$ and
  $d$.

  In $\norm{A}_0$ time, one can compute a randomized integral vector
  $z \in \nnintegers^n$ with coverage
  \begin{math}
    A z \geq \ones,
  \end{math}
  expected cost
  \begin{math}
    \evof{\rip{c}{z}} %
    \leq %
    \parof{\epsmore \ln \columnsum + \bigO{\ln{1/\eps}}} %
    \rip{c}{x},
  \end{math}
  and multiplicities
  \begin{math}
    z \leq \roundup{\epsmore d}.
  \end{math}
\end{corollary}

In \refsection{knapsack-cover-alterations}, we analyze
\algo{contract-round-fix} and prove
\reftheorem{l1-alteration-multiplicity}.

Proofs of \reftheorem{l1-alteration} and
\reftheorem{l1-alteration-multiplicity} can be adapted to obtain
improved approximations when one considers $Ax \ge b$ with
$\bmin > 1$. We prove the stronger bound in \refsection{large-demand}.
We derandomize the algorithm from \reftheorem{l1-alteration} in
\refsection{derandomization}. The other results can be made
deterministic in the same fashion.

\subsection{$\ell_0$-column sparse covering problems}
\labelsection{l0-alteration} In this section, we prove
\reftheorem{l0-alteration}.  The main point of interest is the
expected cost, and the high level approach is as follows. The expected
cost of the rounded solution comes from either the randomized rounding
step or the subsequent alterations, where the expected cost of
randomized rounding is immediate.  To analyze the expected cost of
alterations, we first bound the fixing cost of any constraint $i$,
with cost proportional to the restriction of the fractional solution
$x$ to coordinates $j$ with nonzero coefficients ($A_{i,j} >
0$). Then, we analyze the probability of a constraint $i$ being unmet,
and obtain a probability inversely proportional to $\columncount$. The
expected cost of alteration is the sum, over each constraint $i$, of
the product of probability of failing to meet the $i$th constraint and
the cost of repairing $z$ to fix it. This sum cancels out nicely and
shows that the expected cost of alteration is at most $\rip{c}{x}$.

\begin{lemma}
  \labellemma{basic-concentration} Let $i \in [m]$, and let
  $a = \max_{i,j} A_{i,j}$. After randomized rounding in steps
  \refsubsteps{rounding-loop},
  \begin{align*}
    \probof{(A z)_i < 1} %
    \leq %
    \exp{\frac{1 + \ln{\alpha} - \alpha}{a}}.
  \end{align*}
  For $\alpha \geq \ln \columncount + \ln \ln \columncount + \bigO{1}$,
  we have
  \begin{math}
    \probof{(A z)_i < 1} %
    \leq %
    \frac{1}{2\columncount}.
  \end{math}
\end{lemma}
\begin{proof}
  We apply the Chernoff inequality \reflemma{normalized-chernoff} with
  $\beta = 1$ and $\mu = (A x)_i \geq \alpha$. One can easily verify
  (numerically for instance) that if we choose
  $\alpha = \ln \columncount + \ln \ln \columncount + 4$, then the
  inequality is satisfied for $\columncount \ge 2$.
\end{proof}

\begin{proof}[Proof of \reftheorem{l0-alteration}]
  Let $z \in \nnintegers^n$ be the randomized integral vector output
  by \algo{round-and-fix}.  By the alteration step, we have
  $A z \geq \ones$. We need to bound the expected cost by
  \begin{math}
    \evof{\rip{\cost}{z}} %
    \leq %
    \parof{\alpha + \bigO{1}} \rip{c}{x},
  \end{math}
  and the multiplicities by $z \leq \roundup{\alpha x}$.

  The cost of $z$ comes from the randomized rounding in steps
  \refsubsteps{rounding-loop} and from alterations for unmet
  constraints in steps \refsubsteps{alteration-loop}. The expected
  cost of the rounding step of $\alpha\rip{c}{x}$. For each
  $i \in [m]$, the $i$th constraint is unmet with probability
  $\leq \frac{1}{2\columncount}$ by
  \reflemma{basic-concentration}. When a constraint $i$ is unmet, by
  \reflemma{knapsack-cover}, we can fix $z$ so that $(A z)_i \geq 1$ with
  additional cost at most
  \begin{math}
    2\sum_{j: A_{i,j} \neq 0} c_j x_j.
  \end{math}
  Summed over all $i$, by \tagr interchanging sums and \tagr
  definition of $\columncount$, the expected cost from alteration is
  \begin{align*}
    \frac{1}{2\columncount}\sum_{i =1}^m 2\sum_{j \where A_{i,j} \neq 0}
    c_j x_j                     %
    \tago{=}                           %
    \sum_{j =1}^n c_j \sum_{i \where A_{i,j} \neq 0}
    \frac{x_j}{\columncount}
    \tago{\leq}
    \rip{c}{x},
  \end{align*}
  as desired. Between randomized rounding and alterations, the
  expected cost of $z$ is
  \begin{math}
    \parof{\alpha + 1} \rip{\cost}{x},
  \end{math}
  as desired.

  It remains to bound the multiplicities of $z$. For a fixed
  coordinate $j \in [n]$, $z_j$ is set by either the randomized
  rounding in steps \refsubsteps{rounding-loop} or by an alteration in
  \refstep{constraint-alteration}. The randomized rounding step sets
  $z_j$ to at most $\roundup{\alpha x_j}$. By
  \reflemma{knapsack-cover}, an alteration sets $z_j$ to at most
  $\roundup{2 x_j}$. Thus $z_j \leq \roundup{\alpha x_j}$ for
  $\alpha \ge 2$.
\end{proof}

\subsection{$\ell_1$-column sparse covering programs}

\labelsection{l1-alteration}    %

In this section, we analyze the alteration-based algorithm
\algo{round-and-fix} for covering programs (without multiplicity
constraints) and prove \reftheorem{l1-alteration}. To build some
intuition for our analysis, suppose all of the nonzero coordinates in
$A$ are big (and close to 1). Then we essentially have the
$\columncount$ setting analyzed more simply in
\refsection{l0-alteration}. On the other hand, suppose $A_{i,j} = a$
for all nonzero $A_{i,j}$, for some small $a < 1$. Here the gap
between $\columncount$ and $\columnsum$ is a multiplicative factor of
$\frac{1}{a}$, and the fixing costs from the $\ell_0$-setting sum to
$\bigO{\frac{\columnsum}{a} \rip{\cost}{x}}$. To offset the increased
costs, observe that because the coordinates are uniformly bounded by
$a$, the Chernoff bound tightens exponentially by a
$\frac{1}{a}$-factor, giving a much better bound the probability of
needing to fix each constraint $i$.  In general, the coefficients in a
row $i$ can be non-uniform, with some close to $1$ and some much less
than $1$.  The key to the analysis is identifying a certain threshold
$\score{i}$ for each covering constraint $i$ that divides the
coordinates in the $i$th row between ``big'' and ``small''.  We call
$\score{i}$ the \emph{$i$th (weighted) median coefficient.}
\begin{lemma}
  \labellemma{median}
  Suppose $A x \geq \ones$. For each $i$, there exists
  $\score{i} \in (0,1]$ such that
  \begin{center}
    \begin{math}
      \sum_{j \where A_{i,j} \geq \score{i}} A_{i,j} x_j %
      \geq                                               %
      \frac{1}{2}
    \end{math}
    and
    \begin{math}
      \sum_{j \where A_{i,j} \leq \score{i}} A_{i,j} x_j %
      \geq %
      \frac{1}{2}.
    \end{math}
  \end{center}
\end{lemma}
\begin{proof}
  Let
  \begin{math}
    \score{i} = \inf{ %
      \score %
      \where %
      \sum_{A_{i,j} > \score} A_{i,j} x_j < \frac{1}{2} %
    }.
  \end{math}
  By choice of $\score{i}$, we have
  \begin{math}
    \sum_{A_{i,j} \geq \score{i}} A_{i,j} x_j \geq \frac{1}{2}
  \end{math}
  and
  \begin{align*}
    \sum_{A_{i,j} \leq \score{i}} A_{i,j} x_j %
    = %
    (A x)_i - \sum_{A_{i,j} > \score{i}} A_{i,j} x_j %
    > %
    1 - \frac{1}{2} = \frac{1}{2}.
  \end{align*}
\end{proof}
The next two lemma's consider the fixing cost for constraint $i$,
obtaining a value inversely proportional to $\score{i}$.
\begin{lemma}
  \labellemma{heavy-vector} Let $i \in [m]$.  There exists a vector
  $y \in \nnreals^n$ with coverage
  \begin{math}
    \sum_j A_{i,j} y_j \geq 1,
  \end{math}
  cost
  \begin{math}
    \rip{c}{y} \leq \frac{2}{\score{i}} \sum_j c_j A_{i,j} x_j,
  \end{math}
  and coordinates bounded above by $y \leq 2 x$.
\end{lemma}
\begin{proof}
  Let $y \in \nnreals^n$ be defined by
  \begin{math}
    y_j = 2 x_j
  \end{math}
  if $A_{i,j} \geq \score{i}$, and $y_j = 0$ otherwise. We have $\zeroes
  \leq y \leq 2 x$,
  \begin{math}
    (A y)_i %
    = %
    2 \sum_{j \where A_{i,j} \geq \score{i}} A_{i,j} x_j %
    \geq %
    1
  \end{math}
  by choice of $\score{i}$, and
  \begin{align*}
    \rip{c}{y} %
    = %
    2 \sum_{j \where A_{i,j} \geq \score{i}} c_j x_j %
    \leq                                             %
    \frac{2}{\score{i}} \sum_{j \where A_{i,j} \geq \score{i}} c_j
    A_{i,j} x_j                 %
    \leq                        %
    \frac{2}{\score{i}} \sum_j c_j A_{i,j} x_j,
  \end{align*}
  as desired.
\end{proof}

\begin{lemma}
  \labellemma{l1-fixing-cost} %
  Let $i \in [m]$. In time near-linear in the number of nonzero
  coefficients in row $i$, one can find an integral vector
  $z \subseteq [n]$ with coverage
  \begin{math}
    \sum_{j=1}^n A_{i,j} z_j \geq 1,
  \end{math}
  cost
  \begin{math}
    \sum_{j = 1}^n c_j z_j      %
    \leq                        %
    \frac{4}{\score{i}}\sum_{j=1}^n c_j A_{i,j} x_j,
  \end{math}
  and multiplicities
  \begin{math}
    z \leq \roundup{4 x}.
  \end{math}
\end{lemma}

\begin{proof}
  Applying \reflemma{knapsack-cover} to the vector $y$ of
  \reflemma{heavy-vector} gives the desired result.
\end{proof}

The expected cost incurred from repairing $z$ for the sake of
constraint $i$ is the probability of failing to meet constraint $i$
times the cost given in \reflemma{l1-fixing-cost}.  When $\score{i}$ goes
to zero the expected cost of fixing constraint $i$ is dominated by the
multiplicative factor of $\frac{1}{\score{i}}$. For small $\score{i}$,
we require a stronger concentration bound then
\reflemma{basic-concentration} that decays exponentially in $\alpha$
and \emph{proportionately with $\score{i}$}, to offset the increasing
fixing costs.

\begin{lemma}
  \labellemma{median-concentration} Let $i \in [m]$. After randomized
  rounding in steps \refsubsteps{rounding-loop},
  \begin{align*}
    \probof{(A z)_i < 1}        %
    \leq                        %
    \exp{\frac{1 + \ln{\alpha / 2} - \alpha / 2}{\score{i}}}.
  \end{align*}
\end{lemma}
\begin{proof}
  From \reflemma{median} $\sum_{j:A_{i,j} \le \score{i}} A_{i,j} x_j \ge 1/2$.
  Applying the Chernoff inequality to this sum
  (\reflemma{unscaled-chernoff} with $\gamma = \score{i}$ and
  $\mu = \frac{\alpha}{2}$), we have
  \begin{align*}
    \probof{(A z)_i < 1}        %
    \leq                        %
    \probof{\sum_{A_{i,j} \leq \score{i}} A_{i,j} z_j < 1} %
    \leq                                                   %
    \exp{\frac{1}{\score{i}} \parof{1 + \ln \frac{\alpha}{2} - \frac{\alpha}{2}}},
  \end{align*}
  as desired.
\end{proof}
\begin{lemma}
  \labellemma{concentration}
  Let $i \in [m]$ and
  \begin{math}
    \alpha = \ln \columnsum + \ln \ln \columnsum + \bigO{1}.
  \end{math}
  After randomized rounding in steps \refsubsteps{rounding-loop},
  \begin{math}
    \probof{(A z)_i < 1}        %
    \leq                        %
    C \frac{\score{i}}{\columnsum}
  \end{math}
  for some constant $C > 1$.
\end{lemma}
\begin{proof}
  Let $C > 0$ be a constant to be determined later. We divide the
  analysis into two cases, depending on if $\score{i} \geq \frac{1}{C}$
  or $\score{i} \leq \frac{1}{C}$.

  Suppose $\score{i} \geq \frac{1}{C}$.  In this case we simply use
  \reflemma{basic-concentration} as if for $a = 1$. We have
  \begin{align*}
    \probof{(A z)_i < 1}        %
    \leq                        %
    \frac{1}{\columnsum}          %
    \leq                           %
    C \frac{\score{i}}{\columnsum},
  \end{align*}
  as desired.

  Suppose now that $\score{i} \leq \frac{1}{C}$.  By
  \reflemma{median-concentration}, we have
  \begin{math}
    \probof{(A z)_i < 1} \leq \frac{\score{i}}{\columnsum}
  \end{math}
  if
  \begin{align*}
    \alpha \geq 2 \score{i} \ln{\columnsum} + 2 \score{i}
    \ln{\frac{1}{\score{i}}} + 2 \ln{\alpha / 2} + 2.
  \end{align*}
  For $\alpha = \ln{\columnsum} + \gamma$, as $\gamma \to \infty$ and
  $\score{i} \to 0$, the left side dominates the right side.  In
  particular, for $C$ sufficiently large, and $\gamma$ a sufficiently
  large constant, we have
  \begin{math}                  %
    \probof{(A z)_i < 1} %
    \leq %
    \frac{\score{i}}{\columnsum}. %
  \end{math}
\end{proof}

We conclude the section by completing the proof of
\reftheorem{l1-alteration}.
\begin{proof}[Proof of \reftheorem{l1-alteration}]
  The algorithm ensures that $A z \geq \ones$. The upper bound on the
  multiplicities of $z$ follows by the same argument as in proving
  \reftheorem{l0-alteration}.  It remains to bound the expected cost
  $\evof{\rip{c}{z}}$.

  The expected cost is the sum of the expected cost from the
  randomized rounding in steps \refsubsteps{rounding-loop} and the
  expected cost from alterations in steps
  \refsubsteps{rounding-loop}. The expected cost of the rounding step
  is $\alpha \rip{c}{x}$. For each $i \in [m]$, the expected cost
  incurred by alterations for the $i$th covering constraint is the
  product of the fixing cost for the $i$th constraint, bounded by
  \reflemma{l1-fixing-cost}, and the probability of failing to meet
  the $i$th constraint, bounded by \reflemma{concentration}. Summing
  over all $i$, \tagr interchanging sums, and \tagr definition of
  $\columnsum$, the expected cost of alteration is at most
  \begin{align*}
    C \sum_{i=1}^m \frac{\score{i}}{\columnsum} \cdot
    \frac{\sum_{j=1}^n c_j A_{i,j} x_j}{\score{i}} %
    \tago{=}                                               %
    C \sum_{j=1}^n \frac{c_j x_j}{\columnsum} \sum_{i=1}^m A_{i,j} %
    \tago{\leq}                                                %
    C \rip{c}{x}
  \end{align*}
  for some constant $C \geq 1$.  Summing the expected costs from
  randomized rounding and alterations, the expected cost of the output
  $z$ is $\parof{\alpha + C} \rip{c}{x}$, as desired.
\end{proof}

\subsection{Handling multiplicity constraints}
\input{multiplicity}

\subsection{Derandomization}
\input{derandomization}

\subsection{Improved bound for $\cipi$ when $\bmin$ is large}
\input{large-demand}

\input{mwu}

\paragraph{Acknowledgments:} We thank Neal Young for bringing the paper of Chen, Harris and Srinivasan \cite{chs-16} to our attention.

\bibliographystyle{plainnat} %
\bibliography{cip} %

\appendix                       %

\section{The regime of small $\columnsum$}

\labelappendix{small-columnsum}

\input{small-columnsum}

\input{concentration}

\end{document}

%% file: algos/round-and-fix.tex
\raggedright \ttfamily
\underline{round-and-fix($A \in [0,1]^{m \times n}$,
  $c \in \nnreals^n$, $x \in \nnreals^n$, $\alpha \geq 1$)}
\begin{steps}
  \commentitem{goal: given a fractional point $x \in \nnreals^n$
    with $A x \geq \ones$, output an integral point
    $z \in \nnintegers^n$ with $A z \geq \ones$ and $\rip{c}{z}$
    comparable to $\rip{c}{x}$.}
\item \labelstep{truncate}
  \begin{math}
    z \gets \rounddown{\alpha x},
  \end{math}
  \begin{math}
    x' \gets \alpha x - z
  \end{math}
\item \labelstep{rounding-loop} for $j = 1$ to $n$
  \begin{steps}
  \item with probability $x'_j$
    \begin{steps}
    \item \labelstep{round-coordinate} $z_j \gets z_j + 1$
    \end{steps}
  \end{steps}
\item \labelstep{alteration-loop} for all $i$ such that $(Az)_i < 1$:
  \begin{steps}
    \commentitem{ fix $z$ s.t.\ $(Az)_i \geq 1$}
  \item
    \labelstep{constraint-alteration}
    Find (approximate) solution $y^{(i)}$ for knapsack cover problem
    induced by constraint $i$
  \item $z \gets z \vee y^{(i)}$
  \end{steps}
\item return $z$
\end{steps}


%% file: algos/contract-round-fix.tex
\raggedright \ttfamily
\underline{contract-round-fix($A \in [0,1]^{m \times n}$,
  $c \in \nnreals^n$, $d \in \naturalnumbers^n$, $x \in \nnreals^n$,
  $\alpha \geq 1$, $\eps \in [0,1]$)}
\begin{steps}
  \commentitem{goal: given a fractional point $x \in [\zeroes,d]$
    satisfying all knapsack covering constraints w/r/t the program
    $\setof{Ay \leq \ones \andcomma \zeroes \leq y \leq d}$, output an
    integral point $z \in \nnintegers^n$ with
    $A z \geq \epsless \ones$, $z \leq \roundup{\epsmore d}$, and
    $\rip{c}{z}$ comparable to $\rip{c}{x}$.}
\item
  \begin{math}
    S \gets \setof{j \in [n] \where \alpha x_j \geq d_j},
  \end{math}
  $z'' \gets d \land S$ \commentcode{where
    $(d \land S)_j = d_j$ if $j \in S$ and $0$ otherwise}
\item $b \gets \ones - A z''$
\item $\mathcal{M} \gets \setof{i \in [m] \where b_i > \eps}$,
  $\mathcal{N} \gets [n] \setminus S$
\item define $A' \in [0,1]^{\mathcal{M} \times \mathcal{N}}$
  by
  \begin{math}
    A_{i,j}' = \min{A_{i,j} / b_i, 1}
  \end{math}
\item $z'$ $\gets$ round-and-fix($A'$,$c \land
  \mathcal{N}$,$x \land \mathcal{N}$,$d - z''$)
\item return $z' + z''$.
\end{steps}

%% file: multiplicity.tex
\labelsection{knapsack-cover-alterations}

In this section, we analyze the second algorithm,
\algo{contract-round-fix}. \algo{contract-round-fix} takes as input a
fractional solution that satisfies all the knapsack covering
constraints w/r/t a system of covering constraints $A x \geq \ones$
and multiplicity constraints $x \leq d$ and adds a preprocessing step
to \algo{round-and-fix} to handle multiplicity constraints.

We obtain different approximation factors depending on whether one
desires a pure approximation, meeting the multiplicity constraints
exactly, or a bicriteria factor that approximates the multiplicity
constraints by a $\epsmore$-multiplicative factor. We obtain an
$\ell_0$-sparse bound in the former setting and an $\ell_1$-sparse
bound in the latter.

The algorithms for either setting are similar, with a slight
difference in the choice of parameters $\alpha$ and $\eps$. We adopt
the following common notation when proving either bound.
\begin{itemize}
\item Let
  \begin{math}
    S = \setof{j \in [n] \where \alpha x_j \geq d_j}
  \end{math}
  be the set of coordinates that are deterministically set to their
  multiplicity, and let $z'' = d \land S$ denote the corresponding
  integral vector.
\item Let
  \begin{math}
    \N = [n] \setminus S = \setof{j \where x_j < \alpha}
  \end{math}
  be the
  set of coordinates that are not deterministically rounded up to $d$.
\item Let $\M = \setof{i \where (A z'')_i < 1-\eps}$ be the
  constraints that are not (sufficiently) covered by $z''$.
\item Let $x' = x \land \N$ restrict $x$ to the remaining coordinates
  in $\N$, and let $x'' = x \setminus x' = x \land S$ denote the part
  of $x$ deterministically rounded up.
\item Let $A'$ be the residual covering matrix w/r/t $z''$.
\item Let $z'$ be the random vector output by \algo{round-and-fix}
  w/r/t $A'$ and $x'$.
\item Let $z = z' + z''$ be the combined output.
\end{itemize}

\begin{proof}[Proof of \reftheorem{l0-alteration-multiplicity}]
  Consider \algo{contract-round-fix} with $\eps = 0$ and
  $\alpha = \ln \Delta_0 + \ln \ln \Delta_0 + \bigO{1}$.

  The cost of $z''$ is bounded above by
  \begin{align*}
    \rip{\cost}{z''}            %
    =                           %
    \sum_{j \in S} \cost{j} d_j
    \leq                        %
    \sum_{j \in S} \cost{j} \alpha x_j %
    =                                  %
    \alpha \rip{\cost}{x''},
  \end{align*}
  where $x''$ restricts $x$ to the coordinates in $S$.

  The second vector, $z'$, is the output of \algo{round-and-fix} for
  the residual covering system w/r/t $z''$.  Since $x$ satisfies the
  knapsack covering constraints, and $A' \in [0,1]^{\M \times \N}$ is
  (rescaled) the residual system after deterministically rounding up
  the coordinates in $S$, we have $A' x' \geq \ones$. All the
  coordinates in $A'$ lie between 0 and 1, and the $L_0$ column
  sparsity of $A'$ is bounded above by
  \begin{align*}
    \max_{j \in \N} \sizeof{\setof{i \in \M \where A'_{ij} > 0}} %
    =                                                            %
    \max_{j \in \N} \sizeof{\setof{i \in \M \where A_{ij} > 0}}
    \leq                        %
    \max_{j \in [n]} \sizeof{\setof{i \in [m] \where A_{ij} > 0}} %
    =                                                              %
    \columncount.
  \end{align*}
  By \reftheorem{l0-alteration}, \algo{round-and-fix} returns a vector
  $z''$ with coverage $A'z'' \geq 1$, multiplicities
  $z'' \leq \roundup{\alpha x''} \leq d \land \N$, and expected cost
  \begin{math}
    \evof{\rip{\cost}{z'}} \leq \parof{\alpha + \bigO{1}} \rip{\cost}{x'}.
  \end{math}

  It remains to combine the bounds for $z'$ and $z''$ and analyze
  $z$. The expected cost of $z$ is
  \begin{align*}
    \evof{\rip{\cost}{z}}       %
    =                           %
    \evof{\rip{\cost}{z'}} + \rip{\cost}{z''}
    \leq                        %
    \parof{\alpha + \bigO{1}} \rip{\cost}{x'} + \alpha \rip{\cost}{x''} %
    =                                                %
    \parof{\alpha + \bigO{1}} \rip{\cost}{x}.
  \end{align*}
  For the multiplicity constraints, we have
  \begin{align*}
    z = z' + z'' \leq d \land S + d \land \N = d.
  \end{align*}
  For the coverage, for each $i \in [m]$, we consider two cases
  depending on whether $i \in \M$ or not. If $i \notin \M$, then
  \begin{align*}
    (A z)_i \geq (A z'')_i \geq 1.
  \end{align*}
  If $i \in \M$, then for $b_i = 1 - (A z'')_i$, we have
  \begin{align*}
    (A z)_i = (A z')_i + (A z'')_i
    =                           %
    1 - b_i + b_i (A' z'')_i   %
    \geq                           %
    1 - b_i + b_i                  %
    =                                  %
    1,
  \end{align*}
  as desired.
\end{proof}

Now we consider the bicriteria approximation based on $\ell_1$ sparsity.
\begin{proof}[Proof of \reftheorem{l1-alteration-multiplicity}]
  Consider \algo{contract-round-fix} for a given $\eps > 0$.

  The max column sum of $A'$ is
  \begin{align*}
    \columnsum'
    =
    \max_{j \in \mathcal{N}} \sum_{i \in \mathcal{M}} A_{i,j}' %
    \leq                                                      %
    \max_{j \in \mathcal{N}} \sum_{i \in \mathcal{M}}
    \frac{A_{i,j}}{b_i} %
    \leq               %
    \max_{j \in \mathcal{N}} \frac{1}{\eps} \sum_{i \in \mathcal{M}}
    A_{i,j}
    \leq                        %
    \frac{1}{\eps} \max_{j \in [n]} \sum_{i=1}^m A_{i,j} %
    =                                                    %
    \frac{\columnsum}{\eps}.
  \end{align*}
  By,
  \reftheorem{l1-alteration}, $z'$ has coverage $A'z' \geq \ones$,
  expected cost
  \begin{math}
    \evof{\rip{c}{z'}} \leq \alpha' \rip{c}{x'}, %
  \end{math}
  and multiplicities $z' < \alpha x' + 1$ for
  \begin{math}
    \alpha' = \ln \columnsum' + \ln \ln \columnsum' + \bigO{1}.
  \end{math}
  Observe that, plugging in $\columnsum/\eps$ for $\columnsum'$, we have
  \begin{math}
    \alpha' = \ln \columnsum + \ln \ln \columnsum +
    \bigO{\ln{1/\eps}} %
    = %
    \alpha.
  \end{math}
  Thus, $\evof{\rip{c}{z'}} \leq \alpha \rip{c}{x'}$ and
  $z' \leq \roundup{\alpha x'} \leq d$.

  Then $z''$ has cost
  \begin{math}
    \sum_{j \in S} c_j d_j \leq \alpha \rip{c}{x''},
  \end{math}
  so the expected total cost is
  \begin{align*}
    \evof{\rip{c}{z' + z''}} %
    = %
    \evof{\rip{c}{z'}} + \rip{c}{z''} %
    \leq %
    \alpha \parof{\rip{c}{x'} + \rip{c}{x''}} %
    = %
    \alpha \rip{c}{x}.  %
  \end{align*}
  Since $z'$ and $z''$ have disjoint support, and both $z' \leq d$ and
  $z'' \leq d$ individually, we have $z' + z'' \leq d$. Finally, for
  each constraint $i$, we have either $i \notin \mathcal{M}$, in which
  case
  \begin{math}
    A (z' + z'') \geq A z'' \geq 1 - \eps
  \end{math}
  by definition of $\mathcal{M}$; or $i \in \mathcal{M}$, in which
  case
  \begin{align*}
    A(z' + z'')_i \geq (A z'')_i + (1 - (A z'')_i) (A'z')_i %
    \geq                                               %
    1-b_i + b_i \ge 1,
  \end{align*}
  as desired.
\end{proof}


%% file: derandomization.tex
\labelsection{derandomization}

In this section, we apply the method of conditional expectations to
derandomize the alteration-based rounding schemes for column sparse
covering problems. Recall that the \algo{round-and-fix} algorithm
consists of a randomized rounding step followed by an alteration
step. In particular, randomization only enters when deciding whether
to round a coordinate up or down initially. We thus apply the method
of conditional expectations to the sequence of coin tosses that round
each coordinate up or down.

For ease of exposition, we focus on derandomizing w/r/t $\columnsum$.
Fix a scalar $\alpha > 0$ and fractional solution $x \in \nnreals^n$
with $Ax \geq \ones$. For each $i \in [m]$, we define
$\pessrow{i}{y}: \setof{0,1}^n \to \nnreals$ as the minimum of two
functions,
\begin{math}
  \pessrow{i}{y} \defeq \min{\pessrowA{i}{y}, \pessrowB{i}{y}},
\end{math}
where $\pessrowA{i}: \setof{0,1}^n \to \nnreals$ is defined by
\begin{align*}
  \pessrowA{i}{y}             %
  \defeq                             %
  \alpha \parof{A x}_i %
  \prod_{j=1}^n \parof{\alpha A x}_i^{- \parof{A_{i,j}
  \rounddown{\alpha x_j} + A_{i,j} y_j}}, %
\end{align*}
and
\begin{math}
  \pessrowB{i}: \setof{0,1}^n \to \nnreals
\end{math}
is defined by
\begin{align*}
  \pessrowB{i}{y}                        %
  \defeq
  \prac{\alpha \parof{A x}_i}{2}^{1 / \score{i}} %
  \prod_{j \where A_{i,j} \leq \score{i}} %
  \parof{\frac{\alpha (A x)_i}{2}}^{-\prac{A_{i,j}
  \rounddown{\alpha x_j} + A_{i,j} y_j}{\score{i}}}.
\end{align*}
We define $\pess: \setof{0,1}^n \to \nnreals$ by
\begin{align*}
  \pess{y}                      %
  &=                             %
    \rip{\cost}{\rounddown{\alpha x}}    %
    +                             %
    \rip{\cost}{y}                %
    +                             %
    C \sum_{i=1}^m \pessrow{i}{y}
    \sum_{j=1}^n c_j A_{ij} x_j.
\end{align*}
where $C \in \preals$ is the constant specified by
\reflemma{l1-fixing-cost}. The formula for $\pess{y}$ appears
involved, but has a simple interpretation as a pessimistic estimator
for the expected cost as a function of the randomized rounding in
steps \refsubsteps{rounding-loop}.
\begin{lemma}
  \labellemma{pessimistic-expectation} Let
  \begin{math}
    \alpha = \ln{\Delta_1} + \ln \ln {\Delta_1} + \bigO{1},
  \end{math}
  and let $z \in \nnintegers^n$ be the randomized integral vector
  produced by \algo{round-and-fix}. Let $y \in \setof{0,1}^n$ be the
  random vector with independent coordinates where $y_i = 1$ if $z_i$
  is rounded up in step \refstep{round-coordinate}, and $y_i = 0$
  otherwise. Then
  \begin{align*}
    \evof{\pess{y}} \leq \parof{\alpha + \bigO{1}} \rip{c}{x}.
  \end{align*}
\end{lemma}
\begin{proof}[Proof sketch]
  The claim is implicit in the proof of \reftheorem{l1-alteration}, as
  $\evof{\pess{y}}$ is an intermediate upper bound on
  $\evof{\rip{\cost}{z}}$ in the full chain of inequalities in the
  proof. The expectation of the first two terms,
  \begin{math}
    \rip{c}{\rounddown{\alpha x}} + \evof{\rip{c}{y}},
  \end{math}
  gives the cost from the randomized rounding. For each $i$,
  $\pessrow{i}$ is a pessimistic estimator for the probability that
  $(Az)_i < 1$, occurring in the proof of the Chernoff inequality and
  implicitly bounded from above when we invoked
  \reflemma{basic-concentration} and
  \reflemma{median-concentration}. The full relationship between
  $\evof{\pess{y}}$ and \reftheorem{l1-alteration} is:
  \begin{math}
    \evof{\rip{c}{z}} %
    \leq %
    \evof{\pess{y}} %
    \leq %
    \parof{\alpha + \bigO{1}}\rip{c}{x}. %
  \end{math}
\end{proof}
\begin{lemma}
  \labellemma{pessimistic-cost}
  \begin{math}
    \rip{c}{\rounddown{x}} + \rip{c}{y} + %
    C \sum_{i \where \parof{A \parof{\rounddown{\alpha x} + y}}_i < 1}
    \sum_{j=1}^n c_j A_{i,j} x_j \leq %
    \pess{y}.                          %
  \end{math}
\end{lemma}
\begin{proof}
  Subtracting out common terms, the claim is equivalent to showing
  that
  \begin{align*}
    C \sum_{i \where \parof{A \parof{\rounddown{\alpha x} + y}}_i < 1}
    \sum_{j=1}^n c_j A_{i,j} x_j %
    \leq %
    C \sum_{i=1}^m \pessrow{i}{y} \sum_{j=1}^n c_j A_{i,j} x_j.
  \end{align*}
  Since $\pessrow{i}{y} \geq 0$ for all $i$, it suffices to show that
  $\pessrow{i}{y} \geq 1$ whenever
  \begin{math}
    \parof{A \parof{\rounddown{\alpha x} + y}}_i < 1.
  \end{math}
  Indeed, if this is the case, then
  \begin{align*}
    \pessrowA{i}{y}             %
    =                           %
    \parof{\alpha A x}_i^{      %
    1 - \parof{A \parof{\rounddown{\alpha x} + y}}_i %
    }                                                %
    >                                                %
    \parof{
    \alpha A x
    }_i^0 = 1,
  \end{align*}
  and
  \begin{align*}
    \pessrowB{i}{y}
    =                           %
    \parof{\frac{\alpha A x}{\score{i}}}_i^{      %
    \frac{1 - \parof{A \parof{\rounddown{\alpha x} + y}}_i}{\score{i}}
    }
    >
    \parof{\frac{\alpha A x}{\score{i}}}_i^{0} = 1,
  \end{align*}
  as desired.
\end{proof}

\begin{theorem}
  \labeltheorem{deterministic-l1-alteration} Let $x \in \nnreals^n$
  with $A x \geq \ones$. In nearly linear deterministic time, one can
  compute a vector $z$ with coverage
  \begin{math}
    A z \geq \ones,
  \end{math}
  cost
  \begin{math}
    \rip{\cost}{z} \leq %
    \parof{\ln \Delta_1 + \ln \ln \Delta_1 + \bigO{1}},
  \end{math}
  and multiplicities
  \begin{math}
    z < \alpha x + \ones.
  \end{math}
\end{theorem}
\begin{proof}
  We apply the method of conditional expectations to $\pess{y}$, where
  initially $y \in \setof{0,1}^n$ is a randomized vector with
  independent coordinates and
  $\evof{y} = \alpha x - \rounddown{\alpha x}$. For $j = 1,\dots,m$,
  we fix $y_j$ to either 0 or 1 as to not increase the conditional
  expectation $\evof{\pess{y} \given y_1,\dots,y_j}$.

  Note that $\evof{\pess{y} \given y_1,\dots,y_j}$ is easily
  computable. Moreover, one can arrange a simple data structure such
  that upon advancing $i$ to the next index, the next conditional
  expectation $\evof{\pess{y} \given y_1,\dots,y_j}$ can be recomputed
  in time that is linear in the number of nonzeroes in the $j$th
  column of $A$.

  At the end, by \reflemma{pessimistic-expectation}, we have a fixed
  point $\hat{y} \in \setof{0,1}^n$ such that
  \begin{math}
    \pess{\hat{y}} \leq \evof{\pess{y}} \leq \parof{\alpha + \bigO{1}}
    \rip{c}{x}.
  \end{math}
  We take the integral vector $\rounddown{\alpha x} + \hat{y}$ and fix
  the unmet constraints with \reflemma{l1-fixing-cost}. By
  \reflemma{pessimistic-cost}, this solution has cost at most
  $\pess{\hat{y}} \leq \parof{\alpha + \bigO{1}} \rip{c}{x}$, as
  desired.
\end{proof}


%% file: large-demand.tex
\labelsection{large-demand}

We now consider the case when $Ax \ge b$ where $b_i \ge \bmin > 1$ for
$i \in [m]$. Recall that $\columnsum$ is the maximum column sum of
$A$. We can scale each row $i$ by $b_i$ to obtain a system
$A'x \ge \ones$. If we let $\columnsum'$ denote the maximum column sum
of $A'$, we see that $\columnsum' \le \columnsum/\bmin$.
The analysis that we have already seen would yield an approximation
ration of $\ln \columnsum' + \ln \ln \columnsum' + O(1)$ assuming
that $\columnsum'$ is sufficiently large. In fact one can obtain
a better bound of the form $\frac{\ln \columnsum}{\bmin} + \ln
\frac{\ln \columnsum}{\bmin} + O(1)$. Here we assume that
$\frac{\ln \columnsum}{\bmin}$ is sufficiently large constant.
The analysis closely mimics the one in \refsection{l1-alteration}
and we only highlight the main changes.

The algorithm \algo{round-and-fix} generalizes in the obvious
fashion to the setting when $b \ge \ones$. After the randomized
rounding step each constraint $i$ that is uncovered, that is,
$(Az)_i < b_i$ is greedily fixed by solving a \kncover problem. We
analyze as follows.  \reflemma{median} easily generalizes to show
that for each row $i$, there is a median coefficient $\score{i}$
(w/r/t $x$) such that
\begin{math}
  \sum_{j:A_{i,j} \le \score{i}} A_{i,j}x_j \ge b_i/2
\end{math}
and
\begin{math}
  \sum_{j:A_{i,j} \ge \score{i}} A_{i,j}x_j \ge b_i/2.
\end{math}
\reflemma{l1-fixing-cost} also generalizes to show that the fixing
cost for $i$ is at most
$\frac{4}{\score{i}}\sum_{j=1}^n c_j A_{i,j}x_j$.  We address the
changes to \reflemma{median-concentration} and \reflemma{concentration}.
The probability that constraint $i$ is not covered is:
\begin{align*}
  \probof{(A z)_i < b_i}        %
  \leq                        %
  \exp{\frac{b_i(1 + \ln{\alpha / 2} - \alpha / 2)}{\score{i}}}.
\end{align*}
where we used the Chernoff inequality given by
\reflemma{unnormalized-chernoff} with $\gamma = \score{i}$,
$\mu = \alpha b_i/2$ and $\beta = b_i$.
We also have the following bound.
\begin{align*}
  \probof{(A z)_i < b_i}        %
  \leq                        %
  \exp{b_i(1 + \ln{\alpha} - \alpha)}.
\end{align*}

Our goal is to show that if $\alpha = \frac{\ln \columnsum}{\bmin} + \ln
\frac{\ln \columnsum}{\bmin} + O(1)$ then
$\probof{(A z)_i < b_i}  \le C \frac{\score{i}}{\columnsum}$ for
sufficiently large but fixed constant $C$. We consider two cases as before.
If $\score{i} \ge 1/C$ then
\begin{align*}
  \probof{(A z)_i < b_i}        %
  \leq                        %
  \exp{b_i(1 + \ln{\alpha} - \alpha)} \le \frac{1}{\columnsum} \le C \frac{\score{i}}{\columnsum}.
\end{align*}

Now suppose $\score{i} < 1/C$. We see that $\probof{(A z)_i < b_i}  \le
\score{i}/\columnsum$ if

\begin{align*}
    \alpha \geq 2 \frac{\score{i}}{b_i} \ln{\columnsum} + 2 \frac{\score{i}}{b_i}
    \ln{\frac{1}{\score{i}}} + 2 \ln{\alpha / 2} + 2.
\end{align*}

One can argue as before that with $C$ and the $O(1)$ term in $\alpha$
chosen sufficiently large the inequality holds true.  With these facts
in place, the expected cost of the solution
is $(\alpha + C) \rip{c}{x}$.


%% file: mwu.tex
\section{Fast Algorithm for Solving the Knapsack-Cover LP}
\labelsection{fast}

In this section we develop a fast approximation scheme for
\refequation{kc}. The algorithm is based on speeding up an MWU based
scheme by a combination of technical ingredients.

\subsection{Reviewing the MWU framework and its bottlenecks}
In this section, we give an overview of a width-independent version of
the multiplicative weight update (MWU) framework, as applied to the
dual of \refequation{kc}. Along the way, we review the techniques of
Carr et al.\ \cite{cflp-00} and recover their running time. With some
standard techniques, we improve the running time to nearly quadratic,
and we identify two bottlenecks that we need to overcome to remove the
quadratic factor.

For the remainder of this paper, we assume that
$\eps \geq \frac{1}{\poly{n}}$, since beyond this point one can use
the ellipsoid algorithm instead.

Following \cflp, we apply the MWU framework to the dual of the LP
\refequation{kc}, which is the following pure packing problem.
\begin{align*}
  \begin{aligned}
    \text{maximize } %
    & \sum_{S,i} y_{S,i} b_{S,i} %
    \text{ over } %
    y: \subsetsof{[n]} \times m \to \reals \\
    \text{ s.t.\ } %
    & %
    \sum_{i = 1}^m \sum_{S \subseteq [n]} A_{S,i,j} y_{S,i} \leq
    \cost{j} \text{ for all } j \in [n]
  \end{aligned}
      \labelthisequation[D]{dual}
\end{align*}
Here $\subsetsof{[n]} = \setof{S: S \subseteq [n]}$ denotes the power
set of $[n]$.  The preceding LP has one variable for every constraint
$i$ and every set $S \subseteq [n]$, and corresponds to a single
knapsack covering constraint in the primal LP \refequation{kc}. The LP
\refequation{dual} can be interpretted as packing knapsack covering
constraints into the variables.

The MWU framework is a monotonic and width-independent algorithm that
starts with an empty solution $y = \zeroes$ to the LP
\refequation{dual} and increases $y$ along a sequence of Lagrangian
relaxations to \refequation{dual}. Each Lagrangian relaxation is
designed to steer $y$ away from items $j$ for which the packing
constraint is tight. For each item $j$, the framework maintains a
weight $\weight{j}$ that (approximately) exponentiates the load of the
$j$th constraint with the current solution $y$; i.e.,
\begin{align*}
  \lnof{\cost{j} \weight{j}} \approx       %
  \frac{\log n}{\eps}                      %
  \cdot                %
  \frac{\sum_{i=1}^m \sum_{S \subseteq n} A_{S,i,j}
  y_{S,i}}{\capacity{j}}
  \labelthisequation{weight-update}
\end{align*}
for each $j \in [n]$. Initially, we have
$\weight{j} = \frac{1}{\cost{j}}$ for each $j$. Each iteration, the
framework solves the following Lagrangian relaxation of
\refequation{dual}:
\begin{align*}
  \text{maximize }              %
  &\sum_{S \subseteq [n], i \in [m]} y_{S,i} b_{S,i}   %
    \text{ over } z: \subsetsof{[n]} \times [m] \to \nnreals %
  \\
  \text{ s.t.\ }                %
  &
    \sum_{j=1}^n \weight{j} \sum_{i=1}^m \sum_{S \subseteq [n]} A_{S,i,j} z_{S,i} %
    \leq                                                               %
    \sum_{j=1}^n w_j c_j.
    \labelthisequation[R]{relaxation}
\end{align*}
Observe that the above relaxation biases the solution $z$ away from
items $j$ with large weight $\weight{j}$, which are the items $j$ for
which the packing constraint w/r/t $y$ is tight.  Given an approximate
solution $z$ to the above, we add $\delta z$ to $y$ for a carefully
chosen value $\delta > 0$ (discussed in greater detail below). The
next iteration encounters a different relaxation, where the weights
are increased to account for the loads increased by $z$. Note that the
weights $\weight{j}$ are monotonically increasing over the course of
the algorithm.

At the end of the algorithm, standard analysis shows that the vector
$y$ satisfies $\rip{b}{y} \geq \apxless \opt$ and that $\apxless y$
satisfies all the packing constraints (see for example
\cite{young-14,cjv-15}). The error can be made one-sided by scaling
$y$ up or down. Moreover, it can be shown that at some point in the
algorithm, an easily computable rescaling of $\weight$ is an
$\epspm$-relative approximation for the original LP \refequation{kc}
(see for example \cite{cflp-00,cq-17-focs}). Thus, although we may
appear more interested in solving the dual packing LP
\refequation{dual}, we are approximating the desired LP
\refequation{kc} as well.

The choice of $\delta$ differentiates this ``width-independent'' MWU
framework from other MWU-type algorithms in the literature.  The step
size $\delta$ is chosen small enough that no weight increases by more
than an $\exp{\eps}$-multiplicative factor, and large enough that some
weight increases by (about) an $\exp{\eps}$-multiplicative factor. The
analysis of the MWU framework reveals that
\begin{math}
  \lnof{\rip{\weight}{\cost}} \leq n^{\bigO{1/\eps}}
\end{math}
at all times. In particular, each weight can increase by an
$\exp{\eps}$-multiplicative factor at most
$\bigO{\frac{\ln n}{\eps^2}}$ times, so there are most
$\bigO{\frac{n \ln n}{\eps^2}}$ iterations total.

\subsubsection{Reduction to knapsack cover}

\labelsection{reduction}

An important aspect of the Lagrangian approach is that the
1-constraint packing problem \refequation{relaxation} is much simpler
to solve than the many-constraint packing problem
\refequation{dual}. It suffices to approximately identify the best
bang-for-buck coordinate indexed by $S \subseteq [n]$ and $i \in [m]$;
i.e., approximately maximizing the ratio
\begin{align*}
  \frac{b_{S,i}}{\sum_{S,i,j} w_j A_{S,i,j}},
\end{align*}
and setting $z = \gamma e_{S,i}$ for $\gamma$ as large as possible
within the single packing constraint.  Carr et al.\ \cite{cflp-00}
calls this choice of $S$ and $i$ the ``most violated inequality''.

\cflp reduces the above search problem to a family of knapsack cover
problems as follows. Fix $i \in [m]$.  Expanding out the definitions
of $b_{S,i}$ and $A_{S,i,j}$, finding the set $S$ maximizing the above
is shown to be equivalent to
\begin{align*}
  \text{minimize }              %
  \frac{1}{\alpha}\sum_{j \notin S} \weight{j} \min{A_{i,j},\alpha}
  \text{ over }                 %
  \alpha > 0 \text{ and } S \subseteq [n]
  \text{ s.t.\ }
  \sum_{j \in S} A_{i,j} d_j \leq b_i - \alpha.
\end{align*}
If we let $\excess{i} = \sum_{j=1}^n A_{i,j} d_j - b_i$ denote the
total ``excess'' for the $i$th covering constraint, then we can
rewrite the above as follows:
\begin{align*}
  \text{minimize }
  \frac{1}{\alpha} \sum_{j \in S} \weight{j} \min{A_{i,j},\alpha} %
  \text{ over } \alpha > 0 \text{ and } S \subseteq [n]              %
  \text{ s.t.\ }
  \sum_{j \in S} A_{i,j} d_j \geq \excess{i} + \alpha.
  \labelthisequation[KC]{knapsack-cover}
\end{align*}
For fixed $i$ and $\alpha > 0$, \refequation{knapsack-cover} is a
\emph{knapsack covering} problem.  For the sake of a
$\apxmore$-multiplicative approximation to \refequation{relaxation},
we can approximate the objective by a $\apxmore$-multiplicative
factor, but we must satisfy the covering constraint exactly.

A few basic observations by \cflp allow us to guess $\alpha$ by
exhaustive search.  To obtain a $\epsmore$-multiplicative
approximation to \refequation{relaxation}, we can afford to round
$\alpha$ up to the next integer power of $\epsmore$. Since the nonzero
coefficients $A_{i,j}$ all lie in the range $[1/C, C]$, it suffices to
check $\alpha$ for just $\epslog C = \bigO{\log{C} / \eps}$ powers of
$\epsmore$. We let $\alphas$ denote the set of
$\bigO{\frac{\log C}{\eps}}$ values of $\alpha$ of interest.

A constant factor approximation to \refequation{knapsack-cover} can be
obtained in $\bigO{n_i \log n_i}$ time \cite{cflz-91}, and there are
approximation schemes (discussed in greater detail in
\refsection{kc-oracle}) with running times on the order of
$\apxO{n_i + \poly{1/\eps}}$, where $n_i$ is the number of nonzeroes
in the $i$th row of $A$.  \cflp solves the relaxation
\refequation{relaxation} by applying a FPTAS for knapsack cover to
each choice of $i$ and $\alpha$. For larger values of $m$, by
multiplying the running time of the FPTAS with the number of choices
of $\alpha$ per $i$, and summing over $i \in [m]$, each instance of
\refequation{relaxation} takes
\begin{math}
  \apxO{\frac{N \log C}{\eps} + \frac{m \log C}{\poly{\eps}}}
\end{math}
time to approximate.  In a straight forward implementation of the MWU
framework, each iteration also requires $\bigO{n}$ time to adjust the
weight of each item. With $\bigO{n \log{n} / \eps^2}$ iterations
total, we achieve a running time of
\begin{math}
  \apxO{\frac{n N \log C}{\eps^3} + \frac{m n \log C}{\poly{\eps}}}.
\end{math}
This gives the running time described in \cflp.

\subsubsection{Thresholding}
\labelsection{thresholding}

A standard technique called ``lazy greedy'', ``thresholded greedy'',
or ``lazy bucketing'' in the literature immediately reduces the
running time to
\begin{math}
  \apxO{\parof{N \log C + n^2} \poly{1/\eps}}.
\end{math}
Observe that the optimum ratio is monotonically decreasing as the
weights $\weight{j}$ are monotonically increasing; equivalently, the
optimum value of \refequation{knapsack-cover} is monotonically
increasing for each $\alpha$ and $i$ as the weights $\weight{j}$ are
increasing. This allows us to employ the following thresholding
technique, also used within MWU frameworks in
\cite{fleischer-00,cq-17-soda,cq-17-focs,cq-18}.

We maintain a threshold $\lambda > 0$ such that $\lambda$ is less than
the optimal value of \refequation{knapsack-cover} for all $i \in [m]$
and $\alpha \in \alphas$. The first value of $\lambda$ is obtained by
applying a constant factor approximation algorithm to
\refequation{knapsack-cover} for each $i \in [m]$ and
$\alpha \in \alphas$ and setting $\lambda$ to be a constant factor
less than the minimum cost over all $i \in [m]$ and
$\alpha \in \alphas$. If $\lambda$ is a lower bound for
\refequation{knapsack-cover}, then any solution $S$ and $i$ with ratio
$\leq \apxmore \lambda$ leads to a $\apxmore$-multiplicative
approximation to \refequation{relaxation}. Thus, for a fixed value of
$\lambda$, we solve each $i \in [m]$ and $\alpha \in \alphas$ in
round-robin fashion, taking any $\epsmore$-approximation $S$ with
value $\leq \apxmore \lambda$, or continuing to the next choice of $i$
and $\alpha$ if the returned approximation has value
$\geq \apxmore \lambda$. If all $i \in [m]$ and $\alpha \in \alphas$
generate approximations of value $\geq \apxmore \lambda$, then we can
safely increase $\lambda$ to $\epsmore \lambda$. Observe that since
each weight $\weight{j}$ increases by at most a
$n^{\bigO{1/\eps}}$-multiplicative factor over the entire algorithm,
and then initial choice of $\lambda$ is within a constant factor of
the optimal value for the initial values of $\weight{j}$, $\lambda$
stays within a $n^{\bigO{1/\eps}}$-multiplicative value of its initial
value. In particular, $\lambda$ is bumped up at most
\begin{math}
  \epslog{n^{\bigO{1/\eps}}} = \bigO{\frac{\log n}{\eps^2}}
\end{math}
times.

Each time we approximate an instance of \refequation{kc} for
$i \in [m]$ and $\alpha \in \alphas$, we either (a) find a good
approximation to the relaxation \refequation{relaxation}, or (b)
declare that no solution has value $\leq \epsmore \lambda$ for this
choice of $i$ and $\alpha$ and put the choice of $i$ and $\alpha$
aside until the next value of $\lambda$. That is, each approximated
knapsack cover problem can by charged to either an iteration of the
framework, of which there are $\bigO{\frac{n \log n}{\eps^2}}$, or a
new threshold for this choice of $i$ and $\alpha$. This leads to a
running time on the order of
\begin{math}
  \apxO{\frac{n^2}{\eps^2} + \frac{N \log C}{\eps^3} + \parof{m + n} \poly{\reps}}.
\end{math}

\subsection{Two bottlenecks}
\labelsection{bottlenecks}

Our goal, as stated in \reftheorem{kc-ltas}, is a fast running time on
the order of
\begin{align*}
  \apxO{\frac{N \log C}{\eps^3} + (m+n) \poly{1/\eps}}.
\end{align*}
The bottleneck of $\apxO{\frac{n^2}{\eps^2}}$ appears necessary for at
least two basic reasons. First, there are
$\bigOmega{\frac{n \log n}{\eps^2}}$ iterations, and each iteration
requires a solution $z$ to the relaxation \refequation{relaxation}.
Here $z$ is a vector indexed by $[m]$ and the power set of $[n]$ ---
an $m 2^n$-dimensional space. Even the index of a nonempty coordinate
$(S,i) \in \support{z}$ requires
$\bigOmega{\log{m 2^n}} \geq \bigOmega{n}$ bits to write down. Thus
writing out explicitly a solution to \refequation{relaxation} in each
iteration -- let alone computing it -- generates a
$\bigOmega{\frac{n^2 \log n}{\eps^2}}$ lower bound. Stepping out of
the MWU framework, LP duality tells us that \refequation{dual} can be
minimized by a vector $y$ with support $\sizeof{\support{y}} \leq
n$. If $y$ has $n$ nonzeroes, then even in a sparse explicit
representation of $y$, we need $\bigOmega{n^2}$ to list the supporting
indices. Thus the descriptive complexity of optimal solutions to
\refequation{dual} is seemingly at least $\bigOmega{n^2}$.  A second
bottleneck arises from the weight updates. By the formula
\refequation{weight-update}, the logs of the weights should track the
loads of each packing constraint. However, each iteration may increase
the load of \emph{every} packing constraint, so updating each weight
explicitly can require $\bigOmega{n}$ time per iteration. Thus even
the innocuous weight updates makes a faster running time discouraging.

\subsection{Approximation schemes for knapsack cover}

\labelsection{kc-oracle}

In this section, we review a classical FPTAS for knapsack cover and
set the stage for a more sophisticated integration with the MWU
framework. In the knapsack cover problem, we are given positive costs
and sizes for $n$ items and a positive real-valued size $b$ of a
knapsack; we want to find the minimum (sum) cost subset of items whose
sizes sum to at least the size of the knapsack. In our setting, we
have a sequence of such problems, and the costs and sizes are dictated
via the MWU framework per equation \refequation{knapsack-cover} for
fixed $i \in [m]$ and $\alpha > 0$. The cost of an item $j$ is
\begin{math}
  \kcost{j} \defeq  \frac{\weight{j}}{\alpha} \min{A_{i,j},\alpha},
\end{math}
and the size of an item $j$ is
\begin{math}
  \ksize{j} = A_{i,j} d_j.
\end{math}
For each $j$, $\kcost{j}$ depends linearly on the weight $\weight{j}$
and $\ksize{j}$ is held constant throughout the algorithm.  We want to
update our solution quickly when a weight $w_j$ is increased by the
framework. We are allowed to output solutions that are within a
$\apxmore$-multiplicative factor greater than the optimal objective,
but insist on filling the knapsack completely.

Fix $i \in [m]$ and $\alpha > 0$ and consider equation
\refequation{knapsack-cover}. We let $n_i$ denote the number of
nonzeroes in the row $A_{i,j}$. Ignoring $j \in [n]$ such that
$A_{i,j} = 0$, $n_i$ is the effective number of items in the current
knapsack cover problem.

\begin{figure}
  \small \raggedright
  \begin{framed}
    \ttfamily
    \underline{DP+greedy($\kcost{1},\dots,\kcost{n};\ksize{1},\dots,\ksize{n};b,\beta$)}
    \begin{steps}
      \commentitem{we assume without loss of generality that each item
        has size $\ksize{j} \leq b$, and that $\beta$ is a constant
        factor approximation for $\opt$.}
    \item Compute a collection of $\bigO{\repss}$ pareto optimal sets
      $\mathcal{S}$ over the expensive items (with cost
      $\kcost{j} \geq \eps \beta$) w/r/t the truncated costs
      $\apxkcost{j} = \fracdown{\kcost{j}}{\eps^2 \beta}\eps^2 \beta$.
    \item For each pareto-optimal set $S \in \mathcal{S}$, greedily
      add cheap items (with cost $\kcost{j} < \eps \beta$) to $S$ in
      increasing order of cost-to-size ratio until $S$ fills the
      knapsack.
    \item Return the best solution $S \in \mathcal{S}$
    \end{steps}
  \end{framed}
  \caption{High-level sketch of the algorithm by Lawler
    \cite{lawler-79} for approximating knapsack cover
    problems. \labelfigure{dp+greedy}}
\end{figure}

There are several known approximation schemes for knapsack cover that
run in time $\apxO{n_i + \poly{1/\eps}}$. Here we focus on the similar
approaches of Ibarra and Kim \cite{ik-75} and Lawler
\cite{lawler-79}.\footnote{\cite{ik-75,lawler-79}
  actually consider the more common \emph{maximum knapsack} problem,
  where the goal is to take the maximum sum cost of items that fit
  within the knapsack. Their ideas extend here by using the modified
  greedy algorithm of \cite{cflz-91} to obtain a constant factor
  approximation. Faster or incomparable running times for maximum
  knapsack have been achieved since by more sophisticated techniques
  \cite{kp-04,rhee-15,chan-18}; we follow
  \cite{ik-75,lawler-79} for the sake of simplicity.} A
sketch of the algorithm (following \cite{lawler-79} in
particular) is given in \reffigure{dp+greedy}. The algorithm combines
two basic ideas. First, a greedy heuristic can work fairly well. If
every item has small cost relative to the optimum value, then the
greedy algorithm repeatedly taking the item with minimum cost-to-size
ratio until the knapsack is filled is a good approximation. When costs
are large, the greedy heuristic can be modified to provide a constant
factor approximation within the same $\bigO{n \log n}$ running time
\cite{cflz-91}.  For expensive items (relative to the optimum value),
one can take advantage of the fact that (a) only a few expensive items
can fit in any optimum solution, and that (b) expensive items can have
their costs discretized while changing their costs by only a small
relative factor. After discretization, the expensive items in the
optimum solution can be efficiently guessed by dynamic programming,
and the discretization only introduces a small relative error.

We require the following facts about the algorithm
\algo{DP+greedy}. We state the running time w/r/t the total
number of items in a generic problem, denoted by $n$, which would be
replaced by $n_i$ in our particular setting.
\begin{lemma}[{\cite{ik-75,lawler-79}}]
  \labellemma{ptas-knapsack-cover}
  \begin{enumerate}
  \item If $\beta$ is a constant factor approximation of the optimum
    value, then \algo{DP+greedy} returns a
    $\apxmore$-multiplicative approximation to the minimum cost
    knapsack cover.
  \item If the items are sorted in increasing order of cost-to-size
    ratio, and the prefix sums over the sorted list w/r/t size are
    precomputed, the greedy algorithm can be simulated in
    $\bigO{\log n}$ time.  After preprocessing all the cheap items
    (with cost $\kcost{j} < \eps \beta$) in this way, the greedy
    augmentation in line \algo{3} can be implemented in
    $\bigO{\log n}$ time for each set $S \in \mathcal{S}$.
  \item If, for each $k \in \naturalnumbers$, the expensive items of
    truncated cost
    \begin{math}
      \apxkcost{j} = k \eps^2 \beta
    \end{math}
    are sorted in decreasing order of size, then line \algo{2} can be
    computed in $\bigO{\frac{1}{\eps^4}}$ time.
  \end{enumerate}
\end{lemma}

Marrying \algo{DP+greedy} with the MWU framework in
$\apxO{\poly{\reps}}$ time per iteration has two components. First, we
want to make \algo{DP+greedy} partially dynamic as the costs
$\kcost{j}$ are increased via increments to weights
$\weight{j}$. Second, after computing a good solution, we want to be
able to simulate a weight update along the solution in
$\apxO{\poly{\reps}}$ amortized time. There are basic reasons
(discussed earlier in \refsection{bottlenecks}) why neither component
should be feasible.

\subsection{Dynamically updating the minimum cost knapsack}

\labelsection{kc-dynamic}

The first goal is to be able to respond to weight updates and generate
the solution to the next knapsack cover problem quickly. By
\reflemma{ptas-knapsack-cover}, this boils down to two basic data
structures. First, we need to be able to maintain items in increasing
order of cost-to-size ratio along with the prefix sums w/r/t size in
order to reduce the greedy algorithm to a binary search. This allows
us to greedily augment each candidate set $S \in \mathcal{S}$ with
inexpensive items in line \algo{3} in $\bigO{\log n}$ time per set.
Second, we need to maintain, for each expensive item, all the
expensive items with the same truncated cost in descending order of
size. Note that a constant factor approximation $\beta$ is provided by
the threshold $\lambda$ from the lazy greedy thresholding scheme of
\refsection{thresholding}. Since only a constant factor is required,
we actually set and maintain $\beta$ to be the next power of 2 of
$\lambda$, $\beta = 2^{\logup{\lambda}}$.

We address the second point first, because it is much simpler.  We
need to maintain, for each $k \in \naturalnumbers$, the set of items
with truncated cost $\apxkcost{j} = k \eps^2 \beta$ sorted in
descending order of size. This is very easy --- when an item's cost
increases, we reinsert it into the appropriate sorted list in
$\bigO{n_i}$ time; when $\beta$ increases, we rebuild all the lists,
from scratch, in $\bigO{n_i \log n}$ time.

\begin{lemma}
  \labellemma{dynamic-dp} In $\bigO{\log n}$ time per weight update,
  and $\bigO{n_i \log n}$ time per update to $\beta$, one can
  maintain, for each $k \in \naturalnumbers$, the expensive items of
  truncated cost
  \begin{math}
    \apxkcost{j} = k \eps^2 \beta
  \end{math}
  sorted in descending order of size.
\end{lemma}

The next step is a data structure to facilitate the greedy algorithm
in lines \algo{1} and \algo{3}.  If the items are sorted in ascending
order or cost-to-size ratio, and the prefix sums w/r/t size are
precomputed, then the greedy algorithm can be implemented by a binary
search. In the face of dynamically changing costs, we can maintain
both these values with dynamic tree data structures. Our setting is
simpler because the range of possible costs of a particular item is
known in advance, as follows.
\begin{observation}
  For fixed $j \in [n]$ with $A_{i,j} \neq 0$, let
  \begin{align*}
    \kcosts{j} = \setof{ %
    \frac{\epsmore^k}{\alpha} \min{A_{i,j},\alpha} %
    \where                                         %
    k \in \setof{0,1,\dots,\bigO{\frac{\log m}{\eps^2}} } %
    }.
  \end{align*}
  Then $\kcost{j} \in \kcosts{j}$.
\end{observation}
For $j \in [n]$ with $A_{i,j} \neq 0$, let
\begin{math}
  \kratios{j} = \setof{\alpha / \ksize{j} \where \alpha \in
    \kcosts{j}}
\end{math}
be the set of possible cost-to-size ratios for $j$.  For any $j$ with
$A_{i,j} \neq 0$, we have
\begin{math}
  \sizeof{\kratios{j}} %
  = %
  \sizeof{\kcosts{j}} %
  \leq %
  \bigO{\frac{\log n}{\eps^2}}.
\end{math}

Knowing all the possible ratios in advance allows for a simpler data
structure. We will benefit from the simplicitly later when we need to
incorporate efficient weight updates.  Consider the set
\begin{math}
  \leaves = \setof{ (j, \alpha) %
    \where %
    A_{i,j} \neq 0 %
    \text{ and } %
    \alpha \in \kratios{j} %
  }.
\end{math}
$\leaves$ consists of all possible assignments of ratios to items. We
have $\sizeof{\leaves} \leq \bigO{\frac{n_i \log n}{\eps^2}}$ and
consider $\leaves$ as a sorted set based on the ratio, breaking ties
arbitrarily.

We build a balanced binary tree over $\leaves$. We mark a leaf
$(j,\alpha)$ as occupied iff $\kcost{j} \leq \eps \beta$ and
$\kcost{j} / \ksize{j} = \alpha$. For each internal node, we track the
sum size of all items in leaves marked as occupied.  The tree has
depth $\bigO{\log n}$, and in particular we can update the tree in
$\bigO{\log n}$ time when an item's weight increases. The tree can be
rebuilt $\bigO{n_i \log n}$ time per update to $\beta$.

As the leaves are in ascending order of cost-to-size ratio, any set of
items considered by the greedy algorithm corresponds to the occupied
leaves of an interval in the range tree. The total size of a greedy
set is the total size of occupied leaves in the corresponding
interval.  To compute the sum size of a set of greedily selected
items, we first decompose the corresponding interval in the range tree
to the disjoint union of $\bigO{\log n}$ subtrees. For each subtree,
the root is labeled with the sum size of all occupied leaves in the
subtree. Summing together these $\bigO{\log n}$ values gives us the
total size of the greedy sequence.

With this data structure, we can simulate the greedy algorithm in
$\bigO{\log n}$ time by running a binary search for the shortest
prefix that fills the knapsack.
\begin{lemma}
  \labellemma{dynamic-greedy} In $\bigO{\frac{n_i \log n}{\eps^2}}$
  time initially, $\bigO{\log n}$ time per weight update, and
  $\bigO{n_i \log n}$ time per update to $\beta$, one can maintain a
  data structure that simulates the greedy algorithm over the cheap
  items (with cost $\kcost{j} < \eps \beta$) in $\bigO{\log n}$ time.
\end{lemma}

\subsection{Updating weights along a knapsack cover solution}
\labelsection{kc-weight-update}
In \refsection{kc-dynamic}, we showed how to modify known FPTAS's for
knapsack cover so that it can respond to increases in costs quickly
without redoing everything from scratch. This removes one bottleneck
from the overall MWU framework, in that the Lagrangian relaxations can
now be solved in $\apxO{\poly{1/\eps}}$ amortized time per weight
update.

There is still another bottleneck just as important as computing a
approximation $S$ to the knapsack cover problem; namely, simulating a
weight update w/r/t the solution $S$. Let $S$ be a solution to
\refequation{knapsack-cover} for fixed $i$ and $\alpha$. The MWU
framework dictates that the weights $\weight$ should update to a new
set of weights $\weight'$ per the following formula:
\begin{align*}
  w_j' =
  \begin{cases}
    w_j & \text{if } j \notin S \\
    \exp{ %
      \frac{ %
        \eps A_{S,i,j} / c_j %
      }{ %
        \max_{\varj \in S} A_{S,i,\varj} / c_{\varj} %
      } %
    } w_j %
    & %
    \text{if } j \in S.
  \end{cases}
\end{align*}
The basic problem is as follows. By the above formula, every
coordinate $j \in S$ has its weight adjusted. Updating these weights
directly requires $\bigO{\sizeof{S}}$ time to visit each item. Since
$\sizeof{S}$ may be as large as $n$, and there are
$\bigO{\frac{n \log n}{\eps^2}}$ iterations, this already leads to a
running time of $\apxO{n^2 \log n / \eps^2}$. With no assumptions on
$S$, this is seemingly the best that one can expect.

We circumvent this lower bound by taking advantage of the structure of
our solution $S$. A solution $S$ to the knapsack cover problem
\refequation{knapsack-cover} for fix $i$ and $\alpha$ consists of
$\bigO{\reps}$ expensive items and a greedily selected sequence of
cheap items. Since there are only $\bigO{\reps}$ expensive items, we
can spend $\bigO{\reps}$ time to update each of them individually.  By
contrast, there is no upper bound on the number of cheap items in our
solution.

Recall the data structure by which we simulate the greedy algorithm in
the previous section. The items in greedy order in a balanced binary
tree so that any subsequence of this order decomposes into the leaf
sets of $\bigO{\log n}$ disjoint subtrees. The subtrees implicitly
define \emph{canonical intervals} such that
\begin{enumerate}
\item There are a total of $\bigO{\frac{n_i \log n}{\eps^2}}$
  canonical intervals.
\item Any item appears in at most $\bigO{\log n}$ canonical intervals.
\item Any greedy sequence decomposes into $\bigO{\log n}$ canonical
  intervals.
\end{enumerate}

This is essentially the same setup as in \citep{cq-17-soda}, where
intervals are fractionally packed into capacitated points on the real
line. Similar techniques are also employed in \citep{cq-17-focs}. We
briefly discuss the high-level ideas and refer to previous work
\citep{cq-17-soda,cq-17-focs} for complete details. There is a small
technical adjustment required that is discussed at the end.

Decomposing the solution into a small number of known static sets is
important because weight updates can be simulated over a \emph{fixed}
set efficiently. More precisely, the data structure
\algo{lazy-inc}, defined in \cite{cq-17-soda} and inspired by
techniques by Young \cite{young-14}, simulates a weight update over a
fixed set of weights in such a way that the time can be amortized
against the logarithm of the increase in each of the weights.  The
total increase of any weight is bounded above by invariants revealed
in the analysis of the MWU framework.  It is easy to make
\algo{lazy-inc} dynamic, allowing insertion and deletion into the
underlying set, in $\bigO{\log n}$ time per insertion or deletion
\cite{cq-17-focs}.

We define an instance of \algo{lazy-inc} at each node in the balanced
binary tree over cheap items as defined above. Whenever a leaf is
marked as occupied, it is inserted into each of $\bigO{\log n}$
instances of \algo{lazy-inc} at the ancestors of the leaf; when a leaf
is marked as unoccupied, it is removed from each of these instances as
well. Each instance of \algo{lazy-inc} can then simulate a
weight update over the marked leaves at its nodes in $\bigO{\log n}$
amortized time.

Given a greedy sequence of cheap items, we divide the sequence into
the disjoint union of the marked leaves of $\bigO{\log n}$ subtrees as
discussed above. For each subtree, we simulate a weight update over
the leaves via the instance of \algo{lazy-inc} at the root of
the subtree.

One final technical modification is required to make the algorithm
sound. Each instance of \algo{lazy-inc} accrues a small amount
of error. Within a fixed choice of $i$ and $\alpha$, the sum of errors
for a single weight is small because an item is tracked by only
$\bigO{\log n}$ instances of \algo{lazy-inc}. Across all the
choices of $i$ and $\alpha$, however, a single weight may be managed
by $\bigO{m \log{C} \log{n} / \eps}$ instances of
\algo{lazy-inc}. A similar accrual of error across instances of
\algo{lazy-inc} also arises in \cite{cq-17-focs}.

A final feature of the \algo{lazy-inc} data structure, made explicit
in \cite{cq-17-focs}, is that one can ``flush'' the error of an
instance in $\bigO{1}$ time per tracked item. We use this feature
within the context of the overall lazy greedy algorithm. Recall that
by thresholding the optimum value and trying choices of $i$ and
$\alpha$ in round robin fashion, we move on from a fixed choice of $i$
and $\alpha$ iff the optimum value for these parameters have gone up
by a full $\bigO{1+\eps}$-multiplicative factor, which occurs for any
particular $i$ and $\alpha$ at most $\bigO{\log{n} / \eps^2}$
times. Thus, whenever we move on from a fixed choice of $i$ and
$\alpha$, we ``flush'' all the \algo{lazy-inc} data structures for
this choice of $i$ and $\alpha$ in $\bigO{n_i \log n}$ time, so that
no error is carried across different values of $i$ and $\alpha$.

\begin{lemma}
  \labellemma{kc-weight-update} One can extend the data structures of
  \reflemma{dynamic-greedy} such that, given a solution $S$ generated
  by \reflemma{dynamic-dp} and \reflemma{dynamic-greedy} to an
  instance of \refequation{knapsack-cover} for $i \in [m]$ and
  $\alpha > 0$, one can simulate a weight update over $(i,S)$ in
  $\bigO{\log n}$ amortized time per iteration. The extension induces
  an overhead of $\bigO{\log^2 n}$ per weight update,
  $\bigO{n_i \log n}$ per $\epsmore$-factor increase to $\lambda$, and
  $\bigO{n_i \log^2 n}$ per constant factor increase to $\lambda$.
\end{lemma}

\subsection{Putting things together}

In this section, we summarize the main points of the algorithm and
account for the running time claimed in \reftheorem{kc-ltas}.

\begin{proof}[Proof of \reftheorem{kc-ltas}]
  By known analyses (e.g., \cite[Theorem 2.1]{cq-17-soda}), the MWU
  framework returns a $\apxless$-multiplicative approximation to the
  packing LP \refequation{dual} as long as we can approximate the
  relaxation \refequation{relaxation} in each of
  $\bigO{n \log n / \eps^2}$ iterations. Moreover, each weight
  $\weight{j}$ increases by at most a
  $m^{\bigO{\reps}}$-multiplicative factor, and we can afford to
  approximate $\weight{j}$ by a multiplicative $\apxpm$-multiplicative
  factor and to obtain $\apxpm$-multiplicative approximation to each
  relaxation. Thus, if we only propagate a change to $\weight{j}$ when
  it increases by a $\epsmore$-multiplicative factor, correctness
  still holds, and we can assume that each $\weight{j}$ changes at
  most $\bigO{\log{n} / \eps^2}$ times.

  In \refsection{reduction}, solving \refequation{relaxation} is
  reduced to $\bigO{m \log{C} / \eps}$ instances of minimum knapsack
  cover. The instances of minimum knapsack cover are visited in a
  round robin fashion, where an instance is given a value
  $\lambda > 0$ and only needs to output a set $S$ with ratio
  $\leq \apxmore \lambda$ if there exists a set with ratio
  $\leq \epsmore \lambda$.  For each instance of knapsack cover, we
  maintain an $\epsmore$-approximate solution that outputs an
  approximate solution (in the thresholded sense) in
  $\bigO{\frac{1}{\eps^4}}$ time. By \reflemma{ptas-knapsack-cover},
  \reflemma{dynamic-greedy}, and \reflemma{kc-weight-update}, the data
  structure takes $\bigO{\log n}$ time to update per weight update
  from the framework, $\bigO{n_i \log^2 n}$ time to update every time
  $\lambda$ increases by a power of 2, and $\bigO{n_i \log n}$ time
  every time $\lambda$ increases by a $\epsmore$-multiplicative
  factor. Moreover, if an instance returns a solution within an
  $\apxmore$-multiplicative factor of the threshold $S$, then it can
  simulate a weight update along the corresponding knapsack cover
  constraint $(S,i)$ in $\bigO{\log }$ amortized time (per iteration).

  The threshold $\lambda$ increases by a constant factor
  $\bigO{\log{n} / \eps}$ times and by a $\epsmore$-multiplicative
  factor $\bigO{\log{n} / \eps^2}$ times. The total time spent
  updating the data structure for a single instance of knapsack cover
  for the $i$th constraint is
  \begin{math}
    \bigO{n_i \log^{3} n / \eps + n_i \log^2 n / \eps^2}.
  \end{math}
  The total time over all instances is thus
  \begin{math}
    \bigO{N \log^{3}(n) \log{C} / \eps^2 + N \log^2(n) \log{C} /
      \eps^3}.
  \end{math}
  Querying an instance takes $\bigO{1/\eps^4}$ time, can be charged to
  either an iteration of the MWU framework or bumping the instance up
  to the next threshold. Thus
  \begin{math}
    \bigO{\frac{m \log C}{\eps^5} + \frac{n \log n}{\eps^6}}
  \end{math}
  time is spent querying instances of knapsack cover. The total
  running time is at most
  \begin{math}
    \bigO{\frac{N \log^3(n) \log{C}}{\eps^3} + \frac{m \log C}{\eps^5}
      + \frac{n \log n}{\eps^6}},
  \end{math}
  as desired.
\end{proof}


%% file: small-columnsum.tex
In this section, we consider the regime where $\columnsum$ is
asymptotically small. For $\columnsum$ sufficiently small, we show
that \algo{round-and-fix} obtains an approximation ratio of
$1 + \bigO{\sqrt{\columnsum \ln{1/\columnsum}}}$. Note that this bound
is slightly weaker than the $1 + \bigO{\sqrt{\columnsum}}$ bound
obtained by \citet{chs-16}.

Let $\delta > 0$ be a small parameter to be determined later. (We will
eventually set $\delta = c \sqrt{\columnsum \ln{1/\columnsum}}$ for
some constant $c > 0$.) We set the scaling factor $\alpha$ to
\begin{align*}
  \alpha = \frac{1 + \delta + \delta^2 / {2}}{1 - \delta}.
\end{align*}
Note that
\begin{math}
  \alpha \leq \parof{1 + \parof{1 + \delta} \delta} \parof{1 +
    \delta + \frac{\delta^2}{2}} %
  \leq                           %
  1 + 2\delta + \bigO{\delta^2}
\end{math}
for sufficiently small $\delta > 0$.

The analysis again depends on identifying a ``threshold'' coefficient
$\score{i}$ for each row $i$, but is no longer the weighted median
coordinate (weighted by $x$).  One basic reason that the definition of
$\score{i}$ must be modified is as follows. Recall that considering
only the coverage induced by $\score{i}$-small coefficients allows us
to amplify the Chernoff inequality exponentially by a factor of
$\frac{1}{\score{i}}$. (This in turn offsets the $\frac{1}{\score{i}}$
multiplicative factor in the fixing cost.) To apply the Chernoff
bound, we required that the expected coverage of $\score{i}$-small
(when scaling by up $\alpha$ and rounding) is greater than $1$. In our
case, $\alpha = 1 + \bigO{\delta}$ is very close to $1$. To apply the
Chernoff inequality to $\score{i}$-small coordinates, we require that
the fractional coverage of $\score{i}$-small coordinates is at least
$1 / \alpha$; that is, very close to $1$, rather than just
$1/2$. Thus, for $i$ in $[m]$, we let $\score{i}$ be the
\emph{weighted rank-$\delta$ largest coefficient} as follows.
\begin{lemma}
  For each $i \in [m]$, there exists a value
  $\score{i} \in [0,\columnsum]$ such that
  \begin{align*}
    \sum_{A_{ij} \leq \score{i}} A_{ij} x_j \geq 1 - \delta %
    \text{ and }                                 %
    \sum_{A_{ij} \geq \score{i}} A_{ij} x_j \geq \delta.
  \end{align*}
\end{lemma}

Choosing $\score{i}$ to be the rank-$\delta$ coefficient boosts the
randomized rounding, since the expected coverage from rounding
$\score{i}$-small coordiantes is at least $\parof{1 - \delta}\alpha$.
On the other hand, rounding the $\score{i}$-big coordinates (which
have only $\delta$ fractional coverage) leads to an integral vector
with coverage $\delta$.  To obtain an integral vector with coverage
$1$, one can take $1/\delta$ copies of the integral vector with
coverage $\delta$, increasing the cost by a
$\frac{1}{\delta}$-multiplicative factor.
\begin{lemma}
  \labellemma{small-fixing-cost} For each $i \in [m]$, one can compute
  in near-linear time an integral vector $z \in \nnintegers^n$ with
  cost
  \begin{align*}
    \rip{\cost}{z} \leq
    \bigO{\frac{1}{\delta\score{i}} \sum_{j=1}^n \cost{j}
    A_{ij} x_j}
  \end{align*}
  and coverage
  \begin{math}
    (A z)_i \geq 1.
  \end{math}
\end{lemma}
\begin{proof}
  Consider the problem
  \begin{align*}
    \min \rip{\cost}{z} \text{ over } z \in \nnintegers^n \text{ s.t.\
    } (A z)_i \geq 1.
  \end{align*}
  Observe that each coefficient $A_{ij}$ is $\leq \columnsum$, so this
  is just an instance of \kncover for which there is a PTAS and
  a simple constant factor greedy algorithm, and the integrality gap
  is $2$. Thus it suffices to find a fractional vector
  $y \in \nnreals^n$ with (say)
  \begin{math}
    \rip{\cost}{y} \leq \frac{1}{\delta \score{i}}
    \sum_{j=1}^n \cost{j} A_{ij} x_j
  \end{math}
  and coverage
  \begin{math}
    (A y)_i \geq 1.
  \end{math}

  Define $y \in \nnreals^n$ by
  \begin{align*}
    y_j =
    \begin{cases}
      \frac{1}{\delta} x_j & \text{if } A_{ij} \geq
      \score{i},\\
      0 &\text{otherwise.}
    \end{cases}
  \end{align*}
  First, $y$ has coverage
  \begin{align*}
    \sum_{j} A_{ij} y_j         %
    =                           %
    \frac{1}{\delta} \sum_{A_{ij} \geq \score{i}} A_{ij} x_j
    \tago{\geq} 1,
  \end{align*}
  by \tagr choice of $\score{i}$. Second, $y$ has cost
  \begin{align*}
    \rip{\cost}{y}              %
    =                           %
    \frac{1}{\delta} \sum_{A_{ij} \geq \score{i}} \cost{j}
    x_j %
    \leq                        %
    \frac{1}{\delta \score{i}} \sum_{A_{ij} \geq \score{i}}
    \cost{j} A_{ij} x_j         %
    \leq                        %
    \frac{1}{\delta \score{i}} \sum_j \cost{j} A_{ij} x_j,
  \end{align*}
  as desired.
\end{proof}
\reflemma{small-fixing-cost} establishes an upper bound on the fixing
cost. It remains to obtain concentration bounds for the probability of
the $i$th constraint failing to be covered. Note that we need a
probability of failure proportional to $\score{i}$ for each $i$ to
cancel the $\frac{1}{\score{i}}$ factor in
\reflemma{small-fixing-cost}. In fact, if we can show that
$\probof{(A z)_i < 1} \leq \score{i}$ for all $i$, then the total
expected fixing cost over all $i$ would sum to
$\prac{\columnsum}{\delta} \rip{\cost}{x}$, as desired.  We first
obtain the following.
\begin{lemma}
  \labellemma{small-columnsum-concentration} For each $i \in [m]$,
  $\probof{(A z)_i < 1} \leq e^{-\frac{\delta^2}{2 \score{i}}}$.
\end{lemma}
\begin{proof}
  We apply the Chernoff inequality to the $\score{i}$-small
  coordinates.  From the choice of $\alpha$ and $\score{i}$ we have
  \begin{math}
    \evof{\sum_{j:A_{ij} \le \score{i}} A_{ij}z_j} \geq %
    \frac{\alpha}{1-\delta} %
    \geq %
    1+\delta+\frac{\delta^2}{2}.  %
  \end{math}
  Therefore,
  \begin{align*}
    \probof{(A z)_i < 1} %
    &\leq                            %
      \probof{\sum_{A_{ij} \leq \score{i}} A_{ij} z_j < 1 - \eps} %
    \\
    &\leq                                                            %
      \exp{                     %
      \frac{1}{\score{i}}
      \parof{
      1 + \ln{1 + \delta + \frac{\delta^2}{2}} - 1 - \delta -
      \frac{\delta^2}{2}        %
      }       %
      }\\
    &\tago{\leq}                          %
      \exp{\frac{1}{\score{i}}\parof{1 + \ln{e^{\delta}} - 1 - \delta -
      \frac{\delta^2}{2}}}\\
    &=
      \exp{-\frac{\delta^2}{2\score{i}}}.
  \end{align*}
  by \tagr Taylor expansion of $e^{\delta}$ for $\delta > 0$.
\end{proof}
It remains to choose $\delta$ such that the upper bound in
\reflemma{small-columnsum-concentration} is at most
$\bigO{\score{i}}$. For sufficiently small $\columnsum$, it suffices
to take $\delta \geq 2 \sqrt{\columnsum \ln{1/\columnsum}}$, as follows.
\begin{lemma}
  \labellemma{small-columnsum-concentration-2}
  Let $\delta \geq 2\sqrt{\columnsum \ln{1/\columnsum}}$. For
  sufficiently small $\columnsum$, we have
  \begin{align*}
    \probof{(A z)_i < 1} \leq \score{i}.
  \end{align*}
\end{lemma}
\begin{proof}
  By \reflemma{small-columnsum-concentration}, it suffices to show
  that
  \begin{math}
    \exp{- \frac{\columnsum \ln{1/\columnsum}}{\score{i}}} \leq
    \score{i}.
  \end{math}
  If $\score{i} \leq \columnsum^2$, then we have
  \begin{align*}
    \exp{- \frac{\columnsum \ln{1/\columnsum}}{\score{i}}}
    \leq                        %
    \frac{\score{i}^2}{2 \columnsum^2 \ln^2 \parof{1/\columnsum}}
    \leq                        %
    \score{i},
  \end{align*}
  as desired. If $\score{i} \geq \columnsum^2$, then since
  $\score{i} \leq \columnsum$, we have
  \begin{align*}
    \exp{- \frac{2 \columnsum \ln{1/\columnsum}}{\score{i}}}
    \leq                        %
    \exp{- 2 \ln{1/\columnsum}} %
    =                           %
    \columnsum^2 \leq \score{i},
  \end{align*}
  as desired.
\end{proof}

We now add up the various costs and obtain the desired approximation
factor.
\begin{theorem}
  \algo{round-and-fix} returns a randomized vector $z$ with
  $A z \geq \ones$ and
  \begin{align*}
    \evof{\rip{\cost}{z}} %
    \leq %
    \parof{1 + \parof{4 + o(1)}\sqrt{\columnsum \ln{1 / \columnsum}}} \rip{\cost}{x}.
  \end{align*}
\end{theorem}

\begin{proof}
  Let $\delta = 2 \sqrt{\columnsum \ln{1/\columnsum}}$.
  The total expected cost is sum of the expected cost of randomized rounding,
  \begin{align*}
    \alpha \rip{\cost}{x} = \parof{1 + 2 \delta + \bigO{\delta^2}}
    \rip{\cost}{x}
    =                          %
    \parof{1 + \parof{4 + o(1)} \sqrt{\columnsum
    \ln{1/\columnsum}}} \rip{\cost}{x},
  \end{align*}
  and the total expected fixing cost. The total expected fixing cost
  is at most
  \begin{align*}
    \bigO{\sum_i \probof{(A z)_i < 1} \frac{1}{\delta \score{i}}
    \sum_{j=1}^n c_j A_{ij} x_j} %
    &=                            %
      \bigO{\frac{1}{\delta}\sum_{j=1}^n c_j x_j \sum_{i=1}^m
      \frac{A_{ij} \probof{(A z)_i < 1}}{\score{i}} }
    \\
    &\tago{\leq}                        %
      \bigO{\frac{1}{\delta} \sum_{j=1}^n c_j x_j \sum_{i=1}^m A_{ij}}
          \leq                        %
          \bigO{\frac{\sqrt{\columnsum} \rip{\cost}{x}}{\ln{1/\columnsum}}}
  \end{align*}
  by \tagr \reflemma{small-columnsum-concentration-2}. Thus the
  expected costs of randomized rounding and subsequently fixing add up
  to
  \begin{math}
    \parof{1 + \parof{4 + o(1)} \sqrt{\columnsum
        \ln{1/\columnsum}}}\rip{\cost}{x},
  \end{math}
  as desired.
\end{proof}

\begin{remark}
  Our preceding analysis does not appear to be tight. The expected
  fixing cost is
  \begin{math}
    \bigO{\frac{\sqrt{\columnsum} \rip{\cost}{x}}{\ln{1/\columnsum}}}
\end{math}
for the choice of
\begin{math}
  \alpha = 1 + \bigOmega{\sqrt{\columnsum \log (1/\columnsum)}}.
\end{math}
The two terms in the cost are not balanced.
We note that our analysis did not take into account a more careful
fixing process that considers the expected value of the residual
requirement after the initial random step. Our efforts in doing such a
careful analysis have not so far succeeded in obtaining the bound
$(1 + \bigO{\sqrt{\columnsum}})$ which we believe is the right bound
for the algorithm. In fact it is quite conceivable that the resampling
framework from \cite{chs-16} implies such a bound for the alteration
algorithm.
\end{remark}


%% file: concentration.tex
\section{Chernoff inequalities}

The standard lower tail bound in multiplicative Chernoff inequalities
is the following.

\begin{lemma}
  \labellemma{standard-chernoff}
  Let $X_1,\dots,X_n \in [0,1]$ be independent with
  $\mu = \evof{\sum_{i=1}^n X_i}$, and $\delta \in (0,1)$. Then
  \begin{math}
    \probof{\sum_{i=1}^n X_i < (1-\delta)\mu}
    \leq \parof{\frac{e^{-\delta}}{(1-\delta)^{1-\delta}}}^\mu.
  \end{math}
  If $X_1,\dots,X_n \in [0,\gamma]$ for some $\gamma \le 1$ then
  the bound, by scaling, improves to
    \begin{math}
    \probof{\sum_{i=1}^n X_i < (1-\delta)\mu}
    \leq \parof{\frac{e^{-\delta}}{(1-\delta)^{1-\delta}}}^{\mu/\gamma}.
  \end{math}
\end{lemma}

Typically a simplified version of the upper bound, namely,
$\exp{-\delta^2 \mu/2}$ is used. However, in the interest of leading
constants, we need to employ \reflemma{standard-chernoff} carefully.
For this reason, and also for use as pessimistic estimators for
derandomization in \refsection{derandomization}, we isolate explicit
(intermediate) forms that we need. Several of the following
inequalities are standard and we include proofs for the sake of
completeness.

\begin{lemma}
  \labellemma{moment} %
  Let $X \in \bracketsof{0,1}$ be a random variable and
  $t \in \reals$. Then
  \begin{math}
    \evof{\exp{tZ}} \leq \exp{\evof{Z} (e^t - 1)}.
  \end{math}
\end{lemma}
\begin{proof}
  Consider the function $f(x) = \exp{t x}$. $f$ is convex. By \tagr
  convexity of $f$, for any fixed value of $Z$, we have
  \begin{align*}
    \exp{t Z}                   %
    =                           %
    f((1 - Z) \cdot 0 + Z \cdot 1) %
    \tago{\leq}                           %
    (1-Z) f(0) + Z f(1)            %
    =                              %
    1 - Z + Z e^t                  %
    =                              %
    1 + Z (e^t - 1).
  \end{align*}
  By \tagr taking expectations of both sides and \tagr applying
  $1 + x \leq e^x$, we have
  \begin{align*}
    \evof{t Z}                  %
    \tago{\leq}                        %
    1 + \evof{Z} \parof{e^t - 1} %
    \tago{\leq}                  %
    \exp{\evof{Z} \parof{e^t - 1}}, %
  \end{align*}
  as desired.
\end{proof}

\begin{lemma}
  \labellemma{normalized-chernoff} Let
  $X_1,\dots,X_n \in \bracketsof{0,1}^n$ be independent random
  variables, let $\mu = \evof{\sum_{i=1}^n X_i}$, and
  $\beta \in \preals$. If $\beta < \mu$, then
  \begin{align*}
    \probof{\sum_{i=1}^n X_i < \beta} %
    \leq                              %
    \prac{\mu}{\beta}^{\beta} \prod_{i=1}^n \evof{\prac{\beta}{\mu}^{X_i}}
    \leq %
    \exp{\beta  + \beta \ln \frac{\mu}{\beta} - \mu}.
  \end{align*}
\end{lemma}
\begin{proof}
  This lemma can be proven by appropriate substitutions into the
  standard Chernoff inequality, but we instead prove it directly.  Let
  $t > 0$ be determined later. By \tagr Markov's inequality, \tagr
  independence of the $Y_i$'s, \tagr \reflemma{moment} w/r/t
  $Z = X_i$, we have
  \begin{align*}
    \probof{\sum_{i=1}^n X_i < \beta} %
    &=                                     %
      \probof{\exp{-t\sum_{i=1}^n X_i} > \exp{-t \beta}} %
          \tago{\leq}                      %
          \evof{\exp{-t \sum_{i=1}^n X_i}} \exp{t \beta} %
    \\
    &\tago{=}                   %
      \parof{\prod_{i=1}^n \evof{\exp{-t X_i}}} \exp{t \beta}
          \tago{\leq}
          \parof{\prod_{i=1}^n \exp{\evof{X_i} e^{-t} - 1}} \exp{t \beta}
    \\
    &=                          %
      \exp{\mu \parof{e^{-t} - 1} + t \beta}.
  \end{align*}
  Consider the exponent as a function of $t$,
  \begin{math}
    f(t) = \mu \parof{\exp{-t} - 1} + t \beta.
  \end{math}
  We have
  \begin{align*}
    f'(t) = 0                        %
    \iff                             %
    \mu \exp{-t} =  \beta           %
    \iff %
    t = \ln\frac{\mu}{\beta}.
  \end{align*}
  Consider $t = \ln{\frac{\mu}{\beta}}$.  Note that
  \begin{math}
    \mu > \beta \implies t = \ln \frac{\mu}{\beta} > 0,
  \end{math}
  as required when $t$ was declared above.  For this choice of $t$, we
  have
  \begin{math}
    \probof{\sum_{i=1}^n X_i < \beta} %
    \leq \cdots \leq %
    \exp{\beta - \mu + \beta \ln \frac{\mu}{\beta}}.
  \end{math}
\end{proof}

\begin{lemma}
  \labellemma{unscaled-chernoff} Let $X_1,\dots,X_n \in [0,\gamma]$
  be independent random variables for $\gamma > 0$, let
  $\mu = \evof{\sum_{i=1}^n X_i}$. Then
  \begin{align*}
    \probof{\sum_{i=1}^n X_i < 1} %
    \leq %
    \mu^{1/\gamma} \prod_{i=1}^n \evof{\mu^{-X_i / \gamma}}
    \leq
    \exp{\frac{1}{\gamma} \parof{1 + \ln{\mu} - \mu}}.
  \end{align*}
\end{lemma}
\begin{proof}
  By \tagr scaling up and each $X_i$ by a factor of $1/\gamma$, \tagr
  applying \reflemma{normalized-chernoff} (with $\mu$ scaled up by a
  factor of $1/\gamma$ and $\beta = 1/\gamma$), and \tagr canceling
  terms, we have
  \begin{align*}
    \probof{\sum_{i=1}^n X_i < 1} %
    \tago{=}                      %
    \probof{\sum_{i=1}^n \frac{X_i}{\gamma} < \frac{1}{\gamma}} %
    \tago{\leq}                                                      %
    \exp{\frac{1}{\gamma} + \frac{1}{\gamma} \ln\frac{\mu / \gamma}{1
    / \gamma} - \frac{\mu}{\gamma}} %
    \tago{=}                               %
    \exp{\frac{1}{\gamma}\parof{1 + \ln \mu - \mu}},
  \end{align*}
  as desired.
\end{proof}

\begin{lemma}
  \labellemma{unnormalized-chernoff} Let $X_1,\dots,X_n \in [0,\gamma]$
  be independent random variables for $\gamma > 0$, let
  $\mu = \evof{\sum_{i=1}^n X_i}$, and $\beta \in \preals$. If
  $\beta < \mu$, then
  \begin{align*}
    \probof{\sum_{i=1}^n X_i < \beta} %
    \leq %
    \prac{\mu}{\beta}^{\beta / \gamma} %
    \prod_{i=1}^n \evof{\prac{\beta}{\mu}^{X_i / \gamma}} %
    \leq %
    \exp{\frac{\beta - \mu + \beta \ln{\mu / \beta}}{\gamma}}.
  \end{align*}
\end{lemma}
\begin{proof}
  For each $i$, we have $X_i / \gamma \leq 1$. By \tagr
  \reflemma{normalized-chernoff}, with $\beta$, $\mu$, and the $X_i$'s
  scaled by a factor of $\frac{1}{\gamma}$, we have
  \begin{align*}
    \probof{\sum_{i=1}^n X_i < \beta} %
    = %
    \probof{ %
    \frac{1}{\gamma} \sum_{i=1}^n X_i %
    < %
    \frac{\beta}{\gamma} %
    } %
    \tago{\leq} %
    \exp{\frac{\beta}{\gamma} + \frac{\beta}{\gamma} \ln
    \frac{\mu}{\beta} - \frac{\mu}{\gamma}} %
    =                                       %
    \exp{\frac{1}{\gamma}\parof{\beta + \beta \ln\frac{\mu}{\beta} - \mu}},
  \end{align*}
  as desired.
\end{proof}

\begin{lemma}
  Let $X_1,\dots,X_n \in [0,\gamma]$ be independent with
  $\mu = \evof{\sum_{i=1}^n X_i}$, and $\delta \in (0,1)$. If
  $\mu > \gamma$ and
  \begin{math}
    \frac{\mu}{\gamma} \geq \ln \redelta + \ln \ln \redelta %
    + %
    1 + \frac{1 + \ln \ln{1/\delta}}{\ln{1/\delta} - 1},
  \end{math}
  then
  \begin{math}
    \probof{\sum_{i=1}^n X_i < 1} \leq \delta.
  \end{math}
\end{lemma}
\begin{proof}
  We set
  \begin{math}
    \frac{\mu}{\gamma} = \ln \redelta + \ln \ln \redelta + 1 + \eps
  \end{math}
  for a variable $\eps \in \reals$, and find the minimum the choice of
  $\eps$ that obtains the desired tail inequality.  By
  \reflemma{unscaled-chernoff}, it suffices to choose $\eps$ such that
  \begin{math}
    \frac{\mu}{\gamma} - \ln \frac{\mu}{\gamma} - 1 \geq
    \ln{1/\delta}.
  \end{math}
  We have
  \begin{align*}
    \frac{\mu}{\gamma} - \ln \frac{\mu}{\gamma} - 1 %
    &=                                                %
      \ln \redelta + \ln \ln \redelta + 1 + \eps - \ln{\ln \redelta +
      \ln \ln \redelta + 1 + \eps} - 1
    \\
    &=
      \ln \redelta + \ln \ln \redelta + \eps - \ln{\ln \redelta + \ln
      \ln \redelta + 1 + \eps}.
  \end{align*}
  Set $\eps$ such that
  \begin{math}
    1 + \eps + \ln \ln{1/\delta} = \eps \ln{1/\delta};
  \end{math}
  namely, set
  \begin{math}
    \eps = \frac{1 + \ln \ln{1/\delta}}{\ln{1 / \delta} - 1}.
  \end{math}
  (Note that the requirement that $\delta \in (0,1)$ ensures that
  $\ln{1/\delta} > 1$.)  Then, continuing the above, we have
  \begin{align*}
    \frac{\mu}{\gamma} - \ln \frac{\mu}{\gamma} - 1
    &= \cdots =                 %
      \ln \redelta + \ln \ln \redelta + \eps - \ln{(1+\eps)
      \ln{\redelta}}
    \\
    &=                          %
      \ln{\redelta} + \eps - \ln{1 + \eps} %
      \tago{\geq}                                 %
      \ln{\redelta},
  \end{align*}
  where \tagr uses the inequality $1 + \eps \leq \exp{\eps}$.
\end{proof}

\begin{lemma}
  Let $X_1,\dots,X_n \in [0,1]$ be independent with
  $\mu = \evof{\sum_{i=1}^n X_i}$, $\delta \in (0,1)$, and
  $\beta \in [0,\mu)$. If
  \begin{math}
    \frac{\mu}{\gamma} \geq \ln \redelta + \ln \ln \redelta %
    + %
    1 + \frac{1 + \ln \ln{1/\delta}}{\ln{1/\delta} - 1},
  \end{math}
  then
  \begin{math}
    \probof{\sum_{i=1}^n X_i < \beta} \leq \delta.
  \end{math}
\end{lemma}

\begin{proof}
  By \reflemma{normalized-chernoff}, we have
  \begin{math}
    \probof{\sum_{i=1}^n X_i < \beta} \leq \delta
  \end{math}
  if
  \begin{align*}
    \beta + \beta \ln \frac{\mu}{\beta} - \mu \leq - \ln{\redelta} %
    \iff
    \frac{\mu}{\beta} - \ln \frac{\mu}{\beta} - 1 %
    \geq %
    \frac{\ln{1/\delta}}{\beta}.
  \end{align*}
  Consider setting
  \begin{math}
    \frac{\mu}{\beta} %
    = %
    \frac{\ln{1/\delta}}{\beta} + \ln{\frac{\ln{1/\delta}}{\beta}} + 1
    + \eps
  \end{math}
  for a value $\eps$ to be determined. Then
  \begin{align}
    \frac{\mu}{\beta} - \ln \frac{\mu}{\beta} - 1 %
    -                                             %
    \frac{\ln{1/\delta}}{\beta}                   %
    =                                             %
    \ln{\frac{\ln{1/\delta}}{\beta}} + \eps - \ln{
    \frac{\ln{1/\delta}}{\beta} + \ln{\frac{\ln{1/\delta}}{\beta}} + 1
    + \eps
    }.
  \end{align}
  Set $\eps$ such that
  \begin{align*}
    \ln{\frac{\ln{1/\delta}}{\beta}} + 1
    + \eps
    =                           %
    \frac{\eps \ln{1/\delta}}{\beta};
  \end{align*}
  namely,
  \begin{align*}
    \eps =
    \frac{\ln{\ln{1/\delta} / \beta} + 1}{\ln{1/\delta}/\beta - 1}.
  \end{align*}
  Then
  \begin{align*}
    \reflastequation            %
    =                           %
    \ln{\frac{\ln{1/\delta}}{\beta}}
    + \eps                      %
    -                           %
    \ln{\epsmore \frac{\ln{1/\delta}}{\beta}}
    \tago{=}                           %
    \eps - \ln{1 +\eps}         %
    \geq 0,
  \end{align*}
  where \tagr uses the inequality $1 + \eps \leq e^{\eps}$. Plugging
  back in, we have
  \begin{align*}
    \frac{\mu}{\beta}           %
    =                           %
    \ln{\frac{\ln{1/\delta}}{\beta}} + 1 +
    \frac{\ln{\ln{1/\delta}/\beta} + 1}{\ln{1/\delta}/\beta - 1}
    =                           %
    \frac{\ln{1/\delta}/\beta}{\ln{1/\delta}/\beta - 1}
    \parof{\ln{\frac{\ln{1/\delta}}{\beta}} + 1}
  \end{align*}
\end{proof}
